\documentclass[sn-vancouver,Numbered]{sn-jnl}
\usepackage{graphicx}%
\usepackage{multirow}%
\usepackage{amsmath,amssymb,amsfonts}%
\usepackage{amsthm}%
\usepackage{mathrsfs}%
\usepackage[title]{appendix}%
\usepackage{xcolor}%
\usepackage{textcomp}%
\usepackage{manyfoot}%
\usepackage{booktabs}%
\usepackage{algorithm}%
\usepackage{algorithmicx}%
\usepackage{algpseudocode}%
\usepackage{listings}%
\newtheorem{lemma}{Lemma}

\newtheorem{claim}{Claim}
\newtheorem{corollary}{Corollary}
\theoremstyle{thmstyleone}%
\newtheorem{theorem}{Theorem}%
\usepackage{verbatim}
\theoremstyle{thmstyletwo}%

\theoremstyle{thmstylethree}%

\begin{document}
	
	\title[Article Title]{Short Cycles Decide P-versus-NPC Status of Hamiltonicity on Bisplit Graphs }
	
	
	\author*[1]{\fnm{Mahendra Kumar} \sur{R}}\email{coe18d004@iiitdm.ac.in}
	\author[2]{\fnm{Renjith} \sur{P}}\email{renjith@nitc.ac.in}
	\author[1]{\fnm{Aadhavan} \sur{S}}\email{aadhav1395@gmail.com}
	\author*[1]{\fnm{Sadagopan} \sur{N}}\email{sadagopan@iiitdm.ac.in}

	\affil[1,3,4]{\orgdiv{Computer Science and Engineering}, \orgname{Indian Institute of Information Technology Design and Manufacturing, Kancheepuram}, \orgaddress{ \city{Chennai}, \postcode{600 127}, \state{Tamil Nadu}, \country{India}}}
	
	\affil[2]{\orgdiv{Computer Science and Engineering}, \orgname{National Institute of Technology Calicut}, \orgaddress{ \city{Kozhikode}, \postcode{673 601}, \state{Kerala}, \country{India}}}

	
	\abstract{
		A connected graph $G$ is said to be a bisplit graph if the vertex set of $G$ can be partitioned into a stable set and a complete bipartite graph. We establish the following dichotomy with chordality being the parameter; for chordal bisplit graphs, Hamiltonian cycle (HCYCLE) and Hamiltonian path (HPATH) problems are polynomial-time solvable, and for chordal bipartite bisplit graphs, HCYCLE (HPATH) is NP-complete. We further strengthen the result of \cite{muller1996Hamiltonian} and show that HCYCLE (HPATH) is polynomial-time solvable on  $P_5$-free chordal bipartite graphs (bipartite chain graphs) and NP-complete on $P_{10}$-free chordal bipartite graphs. By using our polynomial results on HCYCLE (HPATH) as a framework, we solve many variants and generalizations of HCYCLE (HPATH), which are also reported in this paper.}


	\maketitle
	
	\section{Introduction}\label{sec1} 
		For a connected graph, the Hamiltonian cycle (path) problem aims to decide whether there is a cycle (path) that visits each node exactly once in the graph. It is one of the classical problems in computing, and it is known that the computational complexity of the Hamiltonian Cycle (HCYCLE) and the Hamiltonian Path (HPATH) problems on split graphs are NP-complete, which are graphs whose vertex set can be partitioned into a clique and a stable set \cite{hpsplit}. Further, many interesting dichotomy results for the Hamiltonian Cycle (Path) problem are reported in \cite{hcsplit,hpsplit}.   Motivated by the work of Renjith et al., in this paper, we ask, what is the computational complexity of  HCYCLE (HPATH)  on bisplit graphs, which are the bipartite analogue of split graphs wherein the vertex set is partitioned into a biclique (complete bipartite graph) and a stable set. The objective of this paper is to perform a fine-grained analysis of HCYCLE (HPATH) in bisplit graphs and discover dichotomy results similar to \cite{hcsplit,hpsplit}.

	A connected graph $G$ is said to be a bisplit graph if its vertex set can be partitioned into a stable set and a complete bipartite graph \cite{bisplit}. Due to its nice structure, it has attracted many researchers working in algorithmic graph theory, and further, it is known that Domination \cite{domi}, Feedback Vertex Set \cite{fvs} are NP-complete on bisplit graphs, and Independent set \cite{golumbie1980algorithmic} and Coloring \cite{color} are polynomial-time solvable in bisplit graphs.

	There is one more motivation for this study.   It is important to highlight that chordal graphs \cite{dirac1961rigid} (graphs that forbid induced cycles of length at least four) and chordal bipartite graphs \cite{fulkerson1972incidence}(bipartite graphs that forbid induced cycles of length at least six),  have a nice structural characterization with respect to minimal vertex separators and a special ordering, namely, perfect elimination ordering \cite{golumbie1980algorithmic} among its vertices (edges). These graphs are the most sought-after graph classes for problems that are known to be NP-complete in general graphs for the reason that these graphs help to understand the gap between NP-complete and polynomial-time solvable instances. Surprisingly, HCYCLE (HPATH) is NP-complete on chordal graphs \cite{bertossi1986Hamiltonian} and chordal bipartite graphs \cite{muller1996Hamiltonian}.   As part of our fine-grained analysis, and to strengthen the existing results on HCYCLE (HPATH), we impose the 'bisplit' property on chordal and chordal bipartite graphs. 
	
	In this paper, we wish to perform a fine-grained analysis of HCYCLE (HPATH) for bisplit graphs with chordality as the parameter. That is, on chordal bisplit and chordal bipartite bisplit graphs. Discovering the interplay between chordality and bisplit property for HCYCLE (HPATH) shall be the primary focus of this paper.   We believe such a study will narrow the gap between the P-versus-NPC boundary in HCYCLE (HPATH). 

Firstly, we show that for chordal bisplit graphs, HCYCLE (HPATH) is polynomial-time solvable, and HCYCLE (HPATH) is NP-complete on chordal bipartite bisplit graphs. These results clearly strengthen the results of \cite{Krishnamoorthy:1975:NPB:990518.990521,munro}. This brings an interesting dichotomy result for HCYCLE (HPATH) on bisplit graphs, with chordality being the parameter. That is, the presence of only $C_3$'s in bisplit graphs, yields a polynomial result for HCYLCE (HPATH), whereas, the presence of only $C_4$'s in bisplit graphs makes the problem NP-complete. Thus, we observe that the 'short cycle' decides the P-versus-NPC status of HCYCLE (HPATH) for bisplit graphs.

To further strengthen the complexity result of HCYCLE (HPATH) on chordal bipartite bisplit graphs, we show that even if the stable set is adjacent to one of the bipartitions of biclique, HCYCLE (HPATH) remains NP-complete. This line of study creates an avenue for further research, and we ask: 

\emph{- What is the complexity status of HCYCLE (HPATH) if convex ordering or nested neighborhood ordering is imposed on one of the partitions?}

 On the other hand, imposing convex property yields a subclass of convex bipartite graphs, which are well-studied graph classes in the literature \cite{hcconvex}, and interestingly, HCYCLE in convex bipartite graphs is polynomial-time solvable, and HPATH in convex bipartite graphs is open (some special cases are polynomial-time solvable).   On the other hand, if we impose nested neighborhood ordering on one of the partitions, we obtain bipartite chain graphs, and further, we establish in this paper that such graphs are $P_5$-free chordal bipartite graphs. We revisit and reprove some of the structural results of $P_5$-free chordal bipartite graphs in a slightly different way with the objective of using our results as a framework to solve many variants and generalizations of HCYCLE (HPATH) restricted to $P_5$-free chordal bipartite graphs. We show that HCYCLE (HPATH) in $P_5$-free chordal bipartite graphs is polynomial-time solvable, and it is NP-complete on $P_{10}$-free chordal bipartite graphs. Our discovery clearly strengthens the results of \cite{muller1996Hamiltonian} and narrows the gap between tractable-versus-intractable instances of chordal bipartite graphs for Hamiltonicity problems.  

Using our framework, we solve Hamiltonicity variants and generalizations in $P_5$-free chordal bipartite graphs, namely Pancyclicity, Hamiltonian connected, Homogeneously traceable, Hypo-Hamiltonian problems, which are known to be NP-complete on chordal bipartite graphs. A connected graph $G$ is pancyclic \cite{bondy1971pancyclic} if $G$ contains cycles of all lengths $l$, $3 \leq l \leq |V(G)|$, and $G$ is said to be bipancyclic \cite{bipancyclic} if $G$ contains a cycles of all even length $2l$, $2 \leq l \leq \frac{|V(G)|}{2}$. A connected graph $G$ is said to be Hamiltonian connected \cite{kewen2008hamiltonianhamiltonconnected} if there exists a Hamiltonian path between every pair of vertices and $G$ is said to be homogeneously traceable \cite{skupien1984homogeneously} if there exists a Hamiltonian path beginning at every vertex of $G$. A hypo-Hamiltonian graph is a non-Hamiltonian graph $G$ such that $G - v$ is Hamiltonian for every vertex $v \in V (G)$. Further, we observe that these variants are NP-complete on $P_{10}$-free chordal bipartite graphs. \\\\	
\noindent\textbf{Organization of the paper:} We next present preliminaries necessary to understand this article. In Section 2, we discuss our hardness results on HCYCLE (HPATH) on chordal bipartite bisplit graphs. The polynomial results on chordal bisplit graphs are discussed in Section 3. In Section 4, we revisit the study on $P_5$-free chordal bipartite graphs and present our framework for Hamiltonicity, its variants, and generalizations. The hardness result of Hamiltonicity for $P_{10}$-free chordal bipartite bisplit graphs is also discussed in Section 4. We conclude this article with some interesting directions for further research.
	\subsection{A few Notation and Definitions}
	In this paper, we work with simple, connected, undirected, and unweighted graphs. We follow the notation and definitions as defined in \cite{golumbie1980algorithmic,dbwest2003}. For a graph $G$, let $V(G)$ denote the vertex set and $E(G)$ denote the edge set. The notation $\{u,v\}$ represents an edge incident on the vertices $u$ and $v$. The neighborhood of a vertex $v$ of $G$, $N_G(v)$, is the set of vertices adjacent to $v$ in $G$. The degree of a vertex $u$ is denoted by $d_G(u)=|N_G(u)|$. We use $d(u)$ and $d_G(u)$ interchangeably if the graph under consideration is clear from the context. The neighborhood of a set $S\subseteq V(G)$ in $G$ is $N_G(S)=\{u~|~u\notin S,v\in S,\{u,v\}\in E(G)\}$.

	A graph $H$ is said to be an induced subgraph of $G$ if for all $u, v \in V (H)$, $\{u, v\} \in E(H)$ if and only if $\{u, v\} \in E(G)$. The graph induced on $V(G)\setminus\{u\}$ is denoted by $G-u$. A path $P$ is a graph on $V(P) = \{u=u_1,u_2,\ldots,v=u_n\}$, $n \geq 1$ and $E(P) =\{\{u_i,u_{i+1}\}~|~1\leq i \leq n-1\}$.  We also denote by $P_{uv}$ or $P_n$, and whose length is $|V(P)|$. A graph $G$ is said to be connected if a path exists between every pair of vertices. A cycle $C$ is a graph on $V(C) = \{u_1,u_2,\ldots,u_n\}$, $n \geq 3$ and $E(C) =\{\{u_i,u_{i+1}\}~|~1\leq i \leq n-1\}\cup \{\{u_n,u_1\}\}$, and $C_k,k\geq3$ denote a cycle on $k$ vertices. A graph $G(X,Y)$ is said to be a biclique if it can be partitioned into two independent sets $X$, and $Y$ such that $E(G)=\{\{u,v\}~|~u\in X, v \in Y\}$, we denote by $K_{m,n}$, $m=|X|, n=|Y|$. For some positive integers $i$ and $j$, a maximal biclique denoted by $K_{i,j}$, is a biclique such that there are no strict supergraphs $K_{i+1,j}$ or $K_{i,j+1}$. A connected graph $G(K,I)$ is said to be a split graph if it can be partitioned into a clique $K$ and an independent set $I$. An edge $e=\{u,v\}$ of a split graph is said to be a clique edge if $u,v\in K$. A connected graph $G(K_1,K_2,I)$ is said to be a bisplit graph if it can be partitioned into a complete bipartite graph ($K_1\cup K_2$ induces a biclique) and an independent set $I$. An edge $e=\{u,v\}$ of a bisplit graph is said to be a biclique edge if $u\in K_1,v\in K_2$. A chordal graph $G$ is said to be strongly chordal if each cycle $C$ of even length has an odd chord in $C$, i.e., an edge joining a pair of vertices in $C$ whose distance is odd. A strongly chordal split graph satisfies both strongly chordal and split properties. For a strongly chordal split graph $G(K,I)$, the underlying bipartite graph $H(X,Y)$ is defined as $X=K$, $Y=I$, and $E(H)=\{\{u,v\}~|~u\in K, v \in I, \{u,v\}\in E(G)\}$. We observe from \cite{muller1996Hamiltonian} that $G(K,I)$ is strongly chordal split if and only if the underlying bipartite graph $H(X,Y)$ of $G(K,I)$ is chordal bipartite.

	\section{Hardness Results: HCYCLE (HPATH) on Chordal Bipartite Bisplit Graphs}
	In this section, we present our first hardness result, which is HCYCLE on chordal bipartite bisplit graphs is NP-complete, which further strengthens the result of \cite{muller1996Hamiltonian}. To establish our result, we present a polynomial-time reduction from HCYCLE in strongly chordal split graphs to HCYCLE in chordal bipartite bisplit graphs.   We observe that the hardness result reported in \cite{muller1996Hamiltonian} for HCYCLE in strongly chordal split graphs is true for the case $|K|=|I|$. However, to present our hardness result, we need the case $|K|>|I|$. We shall first show that on strongly chordal split graphs when  $|K|>|I|$, HCYCLE is NP-complete through Lemma \ref{scpghpe}. To show Lemma \ref{scpghpe}, we use HPATH on chordal bipartite graphs when $|X|=|Y|$ as a candidate problem whose complexity is known to be NP-complete \cite{muller1996Hamiltonian}.\\
	\begin{lemma}\label{scpghpe}
		For strongly chordal split graphs when $|K|=|I|$, HPATH is NP-complete.
	\end{lemma}
	\begin{proof}
		
		We reduce an instance $G(X,Y)$ of HPATH on chordal bipartite graphs when $|X|=|Y|$ to the corresponding instance $H(K,I)$ of HPATH on strongly chordal split graphs when $|K|=|I|$. Our construction is as follows:  $V(H)=K\cup I$; $K=X \cup\{s\}$, $I=Y\cup \{t\}$ and $E(H)= E(G)\cup E'$, $E'=\{\{s,t\}\}\cup \{\{s,v\}\mid v\in K, s\ne v\}$. We now show that $H$ is a strongly chordal split graph. Recall that the underlying bipartite graph of any strongly chordal split graph is chordal bipartite. Observe that in the underlying bipartite graph of $H$, there is no induced cycle of length at least six containing $s$ or $t$. Therefore, $H$ is chordal bipartite. Thus, $H$ is a strongly chordal split graph. Further, it is easy to see that $|K|=|I|$. For the necessary part, if $G$ has a Hamiltonian path $P$, then we extend $P$ by adding $s$, followed by $t$ to get a Hamiltonian path $P'$ in $H$. Conversely, if $H$ has a Hamiltonian path $P'$, then it must be the case that $P'$ starts at $t$ and visits $s$ followed by other vertices of $H$, which are precisely $V(G)$. Clearly, $P'$ without the vertices $s$ and $t$ is a Hamiltonian path $P$ in $G$. \end{proof}
	Using Lemma \ref{scpghpe}, we now prove that on strongly chordal split graphs when $|K|=|I|+1$, HCYCLE is NP-complete.\\

	\begin{lemma} \label{scpghc}
		For strongly chordal split graphs when $|K|=|I|+1$, HCYCLE is NP-complete.
	\end{lemma}
	\begin{proof}
		We present a polynomial-time reduction from HPATH on strongly chordal split graphs when $|K|=|I|$ to HCYCLE on strongly chordal split graphs when $|K|=|I|+1$. Let $G(K,I)$ be an instance of HPATH on the strongly chordal split graph. We now construct the corresponding instance of HPATH on a strongly chordal split graph $H(K',I')$. $V(H)=K'\cup I'$, $K'=K\cup \{s\}$, $I'=I$ and $E(H)= E(G)\cup \{ \{s,v_i\}\mid v_i\in K', s\ne v_i, 1\leq i\leq |K'|\} \cup \{ \{s,u_j\}\mid u_j\in I', 1\leq j\leq |I'|\}$.
		For the necessary part, suppose $G$ has a Hamiltonian path $P$. By our construction, the endpoints of $P$ can be extended to make it adjacent to $s$; thereby, we obtain a Hamiltonian cycle in $H$. Conversely, if $H$ has a Hamiltonian cycle $C$, then $C$ without the vertex $s$ gives the corresponding Hamiltonian path in $G$. 
	\end{proof}

	We now make use of Lemma \ref{scpghc} to prove the hardness results of HCYCLE and HPATH on chordal bipartite bisplit graphs. \\
	\begin{theorem}\label{hcbi}
		For chordal bipartite bisplit graphs, the Hamiltonian cycle problem is NP-complete.
	\end{theorem}

\begin{proof}
	We present a polynomial-time reduction that reduces an instance of HCYCLE on strongly chordal split graphs $G(K,I)$ when $|K|=|I|+1$ to the corresponding chordal bipartite bisplit graph $H(K'_1,K'_2,I'_3,I'_4)$ instance. Let $x_1,x_2,\ldots,x_{|K|}$ represents the vertices of $K$, and $y_1,y_2,\ldots,y_{|I|}$ represents the vertices of $I$ in $G$.  Mapping of $G$ to $H$ is as follows: $V(H)=V_1\cup V_2\cup V_3\cup V_4$; $V_1=\{x^1_i\mid x_i\in K, ~1\leq i \leq |K|\}$, $V_2=\{x^2_i\mid x_i\in K, ~1\leq i \leq |K|\}$, $V_3=\{y^1_j\mid y_j\in I, ~1\leq j \leq |I|\}$, $V_4=\{y^2_j\mid y_j\in I, ~1\leq j \leq |I|\}$ and $E(H)=E'\cup E'', E'=\{\{x^1_i,y^2_j\},\{x^2_i,y^1_j\}\mid  \{x_i,y_j\}\in E(G)\}$, $E''=\{\{x^1_i,x^2_j\}\mid x^1_i\in V_1,~x^2_j\in V_2, ~1\leq i \leq |K|,1\leq j \leq |K|\}$.  Note that $H$ is a bisplit graph with $V_1\cup V_2$ as biclique and $V_3 \cup V_4$ as an independent set. We now show that $H$ is chordal bipartite. Suppose that there exists an induced cycle of length at least six $C_l$, $l\geq 6$. Since $V_1\cup V_2$ is a biclique, $C_l$ can contain at least two vertices from $V_1$ and at most one vertex from $V_2$ or vice versa. Without loss of generality, assume that $C_l$ contains at least two vertices from $V_1$ and at most one vertex from $V_2$. This implies that at least five vertices must be from  $V_1 \cup V_4$. Clearly, this shows that the underlying bipartite graph of $G$ is not a chordal bipartite, a contradiction. Therefore, $C_l$, $l\geq 6$ does not exist in $H$. Thus, $H$ is chordal bipartite. We claim that $G$ is a yes-instance of HCYCLE if and only if $H$ is a yes-instance of HPATH.
	
	\emph{{Necessity:}} Suppose $G$ has a Hamiltonian $C$. Let $E(C)$ denote the edge set of $C$.  Since $|K|=|I|+1$, $C$ must have exactly one clique edge $e=\{x_i,x_j\}$.   Clearly, $C$ without $e$ is a path $P_{x_ix_j}$ spanning $V(G)$. For simplicity, we use $P$ to refer $P_{x_ix_j}$. We use the same order in which the path $P$ visits the vertices of $G$ to construct a subpath $P'=(x^1_i,\ldots x^1_j)$ that spans $V_1\cup V_4$. An edge $\{x_l,y_s\}$ in $P$ is modified to the corresponding edge $\{x^1_l,y^2_s\}$ in $P'$. Similarly, we obtain a subpath $P''=(x^2_i,\ldots x^2_j)$ that spans $V_2\cup V_3$ using. Observe that the endpoints of $P'$ are from $V_1$ and the endpoints of $P''$ are from $V_2$. Since $V_1\cup V_2$ is a biclique, the endpoints of $P'$ and $P''$ can be extended to obtain a cycle  $C'$ in $H$. The constructed $C'$ is as follows; $C'=(x^1_i,\ldots,x^1_j,x^2_i,\ldots,x^2_j,x^1_i)$.   
	
	\emph{{Sufficiency:}}  Suppose $H$ contains a Hamiltonian cycle $C'$, we now construct the corresponding Hamiltonian cycle $C$ in $G$. Let $X_1=V_1 \cup V_4$ and $X_2=V_2 \cup V_3$. Observe that any Hamiltonian cycle  must use  edges of type $\{x^1_i,x^2_j\}$ or $\{x^1_i,x^2_i\}$ (biclique edges) to alternate between $X_1$ and $X_2$.  Since $|V_1| = |V_4|+1$ and $|V_2| = |V_3|+1$, the number of biclique edges is bounded by two, say $e_1$ and $e_2$. Clearly, $C'$ without $e_1$ and $e_2$ gives two subpaths $P'=(x^1_i,\ldots,x^1_j)$ and $P''=(x^2_i,\ldots,x^2_j)$ such that $P'$ spans the vertices of $X_1$ and $P''$ spans the vertices of $X_2$. Without loss of generality, we consider $P'$ to construct the corresponding Hamiltonian cycle $C$ in $G$.  By our construction the graph induced on $X_1$ is same as $G$ and $\{x_i,x_j\}\in E(G)$. Thus, $P'$ along with $\{x_i,x_j\}$ is the desired Hamiltonian cycle in $G$.  	
\end{proof}
	
	\noindent
	\textbf{Insights into the reduction instances of Theorem \ref{hcbi}:} 
	We observe that the diameter of a bisplit graph is at most four. Since $H$ has two universal vertices (one at each partition), the reduction instances are diameter three chordal bipartite bisplit graphs. Thus, we conclude that for chordal bipartite bisplit graphs with diameter three, HCYCLE is NP-complete.   Further, the $H$ has an induced path $P_k$ of arbitrary length. Thus, we ask, what is the complexity status of HCYCLE in $P_k$-free chordal bipartite bisplit graphs, $k\geq 5$. We shall revisit this question in Section \ref{p5}. \\
	\begin{theorem} \label{hpbi}
		For chordal bipartite bisplit graphs, the Hamiltonian path problem is NP-complete.
	\end{theorem}
	\begin{proof}
		The reduction is from the Hamiltonian cycle problem on chordal bipartite bisplit graphs. We present a polynomial-time reduction that reduces an instance of a chordal bipartite bisplit graph $G(K_1,K_2,I)$ to the corresponding chordal bipartite bisplit graph $H(K'_1,K'_2,I')$ instance. We map an instance of $G$ to the corresponding instance of $H$ as follows: Let $w\in K_1$. We create a vertex $w'$ in $K'_1$, which is a copy of $w$.  $V(H)=K'_1\cup K'_2 \cup I'$; $K'_1=K_1\cup \{w'\}$, $K'_2=K_2$, $I'=I\cup \{s,t\}$. Let $w\in K_1$, and let $S=N_G(w)$. $E(H)=E(G)\cup E'\cup \{\{w,s\},\{w',t\}\}$, $E'=\{\{w',z\}\mid z\in S\}$. It is easy to see that $H$ is a bisplit graph. Since $w'\in K'_1$ and $K'_1\cup K'_2$ is a biclique, $w$ does not create a cycle of length at least six in $H$. Thus, $H$ is chordal bipartite and bisplit. We claim that $G$ is a yes-instance of the Hamiltonian cycle problem if and only if $H$ is a yes-instance of the Hamiltonian path problem.
		
		For the necessary part, suppose $G$ has a Hamiltonian cycle $C$. Let $C=(w,x,\ldots,x',w)$. Since $N_G(w)=N_G(w')$, $\{x',w'\}\in E(H)$. Clearly, $(s,w,x,\ldots,x',w',t)$ is a Hamiltonian path in $H$. Conversely, since $s$ and $t$ are pendant vertices, any Hamiltonian path $P$ in $H$ must start at $s$ and end at $t$ and vice versa. Let $P=(s,w,x,\ldots,x',w',t)$. Observe that $\{x',w\}\in E(G)$. Clearly, $P$ without $s,t$ and $w'$ and along with the edge $\{x',w\}$ is a Hamiltonian cycle in $G$.
	\end{proof}	 
	\noindent
	{\bf Remarks:} Results on chordal bipartite bisplit graphs reveal that the presence of only $C_4$'s makes the computation of HCYCLE (HPATH) difficult. Further, we know from the result of Dhanalakshmi et al. \cite{dhana}, that on bisplit graphs with only $C_k$, for fixed $k \geq 6$, computation of HCYCLE (HPATH) is polynomial-time solvable (strictly chordality $k$ graphs).   It is now natural to ask what is the computational complexity of HCYCLE (HPATH) in chordal bisplit graphs (bisplit with  only $C_3$'s), which we shall address next.
	\section{Hamiltonian cycle (path) in Chordal Bisplit Graphs}
	In this section, we shall investigate the structure of chordal bisplit graphs, which are bisplit graphs and satisfy the chordality property, i.e., the graph is either acyclic or any induced cycle is of length three. Using our structural results, we present a polynomial-time algorithm for HCYCLE (HPATH) in chordal bisplit graphs. For a chordal bisplit graph $G$ with bipartition $K_1,K_2,I$, let $K_1=\{x_1,x_2,\ldots,x_m\}$, $K_2=\{y_1,y_2,\ldots,y_n\}$ and $I=\{z_1,z_2,\ldots,z_l\}$ where $K_1\cup K_2$ is a biclique and $I$ is an independent set. Observe that for chordal bisplit graphs with bipartition $K_1,K_2,I$, if both $|K_1|\geq 2$ and $|K_2|\geq 2$ then there is an induced cycle of length four. Therefore, one of them, say, $|K_1|\leq 1$. This implies that the structure of the underlying complete bipartite graph is a star. Let $\{x,y_1,y_2,\ldots, y_n\}$ induces a star with $x$ as the root of the star and $\{y_1,y_2,\ldots, y_n\}$ as the leaves.  To compute HCYCLE (HPATH), we have to understand the structure of the graph induced on $K_2\cup I$. We shall next present the structure of $K_2\cup I$ for the existence of HCYCLE (HPATH).\\
\begin{lemma}\label{degree}
	Let $G=(K_1=\{x\},K_2,I)$ be a chordal bisplit graph. For all $z\in I$, if $|N_G(z) \cap K_2|\geq2$, then $\{z,x\}\in E(G)$.
\end{lemma}
\begin{proof}
	Suppose that there exists a vertex $z\in I$ such that $|N_G(z) \cap K_2|\geq2$ and $\{z,x\}\notin E(G)$.   Then, for some $y,w \in (N_G(z) \cap K_2)$, $\{z,y,w,x\}$ induces a $C_4$,  contradicting the chordality property of $G$.   \end{proof}
\begin{lemma}\label{k2structure}
	Let $G=(K_1=\{x\},K_2,I)$ be a connected chordal bisplit graph. Then, the graph induced on $K_2\cup I$ is either a tree or a forest.
\end{lemma}
\begin{proof} Since $G$ is chordal, and $K_2\cup I$ induces a biparite graph, it follows that $K_2\cup I$ is either a tree or a forest. 
\end{proof}
\begin{theorem}\label{HCBCG}
	Let $G=(K_1=\{x\},K_2,I)$ be a connected chordal bisplit graph. Then $G$ has a Hamiltonian cycle if and only if the graph induced on $K_2\cup I$ is a path $P$ whose endpoints are adjacent to $x$.
\end{theorem}
\begin{proof}
	Let $C$ be a Hamiltonian cycle in $G$. It is clear from Lemma \ref{k2structure} that the graph induced on $K_2\cup I$ is either a tree or a forest. Since $G$ has a Hamiltonian cycle, it must be the case that the graph induced on $K_2\cup I$ is a path. This implies that $C$ without $x$ is a path spanning $K_2\cup I$ with the property that the endpoints of $P$ are adjacent to $x$. Conversely, $P$ is a path on $K_2\cup I$ whose endpoints are adjacent to $x$. Clearly, $P$ along with $x$, is the desired Hamiltonian cycle in $G$.
\end{proof}

\noindent We shall now analyze the structure of $K_2\cup I$ to show the existence of a Hamiltonian path. We observe that the structure of $K_2\cup I$ is a special tree of Type- $H_1,H_2,H_3,H_4$ which are illustrated in  Figure \ref{hpathcb}.\\
\begin{figure} 
	\begin{center}
		\includegraphics[width=130mm,scale=0.5]{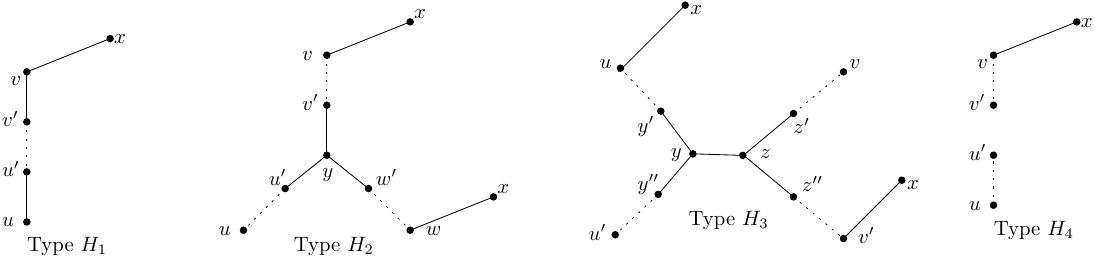}
		\caption{\small \sl   Possible structures of $K_2\cup I$ for the existence of a Hamiltonian path \label{hpathcb}}
	\end{center}
\end{figure}
\begin{theorem}\label{hpcbsp}
	Let $G=(K_1=\{x\},K_2,I)$ be a chordal bisplit graph. Then, $G$ has a Hamiltonian path if and only if the graph induced on $K_2\cup I$ is isomorphic to $H_1, H_2$, $H_3$, or $H_4$.
\end{theorem}
\begin{proof}
	\emph{{Necessity:}} Let $P$ be a Hamiltonian path in $G$. Since $P$ is a spanning path, $x$ must appear somewhere in $P$. Let $P=(\ldots,x',x,x'',\ldots)$, and $x',x''$ be the neighbors of $x$ in $P$. We now construct two paths from $P$, $P'=(\ldots,x')$ and  $P''=(x'',\ldots)$.  Note that the vertices of $P'$ and $P''$ are from $K_2\cap I$. From Lemma \ref{k2structure}, it is clear that the graph $H$ induced on $K_2\cup I$ is either a tree or a forest. Observe that in $H$, the endpoints of $P'$ and $P''$ can be pendant vertices. This shows that the number of pendant vertices in $H$ is bounded, which is at most four. We shall now show that the graph $H$ is isomorphic to any one of $H_1, H_2$, $H_3$, or $H_4$. Case 1: $H$ is a tree. Since $H$ is connected, there must exist exactly one edge $e=\{s,t\}$ in $H$ between a vertex of $P'$ and a vertex of $P''$,  $s \in V(P'), t \in V(P'')$. If both $s$ and $t$ are endpoints, then $H$ isomorphic to $H_1$.  Similarly, if $s$ is an endpoint and $t$ is an intermediate vertex, then $H$ is isomorphic to $H_2$. Suppose both $s$ and $t$ are intermediate vertices, then $H$ is isomorphic to $H_3$. Case 2: $H$ is a forest. Since $H$ is disconnected, an edge $e$ in $H$ that connects the vertices of $P'$ and $P''$ does not exist. This implies that $H$ is a graph on two disconnected paths. Thus, $H$ is isomorphic to $H_4$. 
	
	Conversely, suppose $K_2\cup I$ is isomorphic to any one of (i) $H_1$, (ii) $H_2$, (iii) $H_3$, or (iv) $H_4$.	We construct the Hamiltonian path $P$ as follows. (i) $P=(x,v,v',\ldots,u',u)$. (ii) $P=(u, u',y,v',\ldots,v,x,w,\ldots,w')$. (iii) $P=(u',\ldots,y'',y,y',\ldots,u,x,v',\ldots,z'',z,z',\ldots,v)$. (iv) $P=(v',\ldots, v,x,u,\ldots,u')$. 
\end{proof}
\noindent	\textbf{Remarks:} Theorem \ref{HCBCG} and Theorem \ref{hpcbsp} are constructive in nature. Therefore, the Hamiltonian cycle (path) in chordal bisplit graphs can be obtained in polynomial time. Interestingly, chordal bisplit graphs, when they are Hamiltonian, also have all possible cycles of length $3$ to $n$. Such graphs are known as pancyclic graphs in the literature \cite{bondy1971pancyclic}.\\
\begin{theorem}
	For a chordal bisplit graph $G=(K_1=\{x\},K_2,I)$ with $n$ vertices, $G$ is pancyclic if and only if $G$ is Hamiltonian.
\end{theorem}
\begin{proof}
	Since $G$ is pancyclic, $G$ contains cycles of all lengths, including $n$, and hence $G$ is Hamiltonian. Conversely, since $G$ has a Hamiltonian cycle, as per Theorem \ref{HCBCG}, $K_2\cup I$ is a path $P$ whose endpoints are adjacent to $x$. Further, for all $z\in I$ that are the internal vertices of $P$, $|N_G(z)\cap K_2|\geq 2$. This shows that for all $z \in I$, $\{x,z\}\in E(G)$ (Lemma \ref{degree}). Recall that $\{x\} \cup K_2$ induces a star with $x$ as the root in $G$. Clearly, $x$ is adjacent to all the vertices of $P$. To construct a cycle $C_i$, we consider a subpath $P_{i-1}$ on first $i-1$ vertices of $P$ and $x$. \\
	$C_3=(x,P_2,x)$\\
	$C_4=(x,P_3,x)$\\
	$\vdots$\\
	$C_{n-1}=(x,P_{n-2},x)$\\
	$C_{n}=(x,P_{n-1},x)$\\
	Clearly, we obtain all possible cycles in $G$. Thus $G$ is pancyclic.
\end{proof}
	\begin{theorem}
		\label{thmhomogeneouslychordal}
		For a chordal bisplit graph $G=(K_1=\{x\},K_2,I)$, $G$ is Hamiltonian if and only if $G$ is homogeneously traceable.
	\end{theorem}
	\begin{proof}
		It is easy to see that the Hamiltonian graphs are homogeneously traceable. Conversely, suppose $G$ is homogeneously traceable.   Since $G$ has a Hamiltonian path that starts at every vertex, it must be the case that the graph induced on $K_2\cup I$ is isomorphic to $H_1$ (mentioned in Theorem \ref{hpcbsp}). Further, $x$ must be universal to $K_2\cup I$. Since $K_2\cup I$ is a path $P=(x,v,v',\ldots,u',u,)$, $P$ with the edge $\{x,u\}$ is a Hamiltonian cycle in $G$. 
	\end{proof}

	\section{Reducing P-vs-NPC gap even further}
	Having seen the impact of chordality on bisplit graphs for Hamiltoncity, the presence of only $C_4$'s makes the problem hard, and the presence of only $C_k$'s, for any fixed $k$ other than four, makes the problem easy.  It is natural to investigate NP-complete instances (chordal bipartite bisplit) even further to understand the property or structure that makes the problem difficult.   We next show that (i) if there is no ordering on $K_1$ (one of the partitions) and all of $I$ is adjacent to only $K_1$, then the Hamiltonicity remains NP-complete, however, (ii) if we impose nested neighborhood ordering on $K_1$ ($K_2$), then the problem becomes polynomial-time solvable.   Interestingly, bisplit graphs with nested neighborhood ordering have the property that they are $P_5$-free chordal bipartite (also known as bipartite chain graphs).   In this section, we narrow the complexity gap further and show that Hamiltonicity on $P_{10}$-free chordal bipartite graphs is NP-complete and polynomial-time solvable on $P_5$-free chordal bipartite graphs.  
	\subsection{Another Hardness Result on chordal bipartite bisplit}
	In this section, we show that HCYCLE is NP-complete on chordal bipartite bisplit if no ordering on $K_1$ and all of $I$ is adjacent to only $K_1$.\\
	\begin{theorem}\label{hcbioneside}
		Let $G(K_1,K_2,I)$ be a chordal bipartite bisplit graph such that $N(I)\cap K_2=\emptyset$ and no ordering on $K_1$. Then, the Hamiltonian cycle problem is NP-complete.
\end{theorem}
\begin{proof}
We reduce an instance $G(K,I)$ of HCYCLE on strongly chordal split graph when $|K|=|I|+1$ to the corresponding instance $H(K'_1,K'_2,I')$ of HCYCLE on chordal bipartite bisplit graph such that $G$ has a Hamiltonian cycle if and only if $H$ has a Hamiltonian cycle. We construct $H$ as follows. Let $x_1,x_2,\ldots,x_{|K|}$ represents the vertices of $K$, and $y_1,y_2,\ldots,y_{|I|}$ represents the vertices of $I$ in $G$.  We shall now define $V(H)$,  $V(H)=V_1\cup V_2\cup V_3$; $V_1=\{x^1_i\mid x_i\in K, ~1\leq i \leq |K|\} \cup \{x^1_{|K|+i}\mid 1\leq i\leq
|K|-1\}$, $V_2=\{x^2_i\mid x_i\in K, ~1\leq i \leq |K|\}$, $V_3=\{y^1_j\mid y_j\in I, ~1\leq j \leq |I|\}$, and $E(H)=E'\cup E'', E'=\{\{x^1_i,y^1_j\}\mid  \{x_i,y_j\}\in E(G)\}$, $E''=\{\{x^1_i,x^2_j\}\mid x^1_i\in V_1,~x^2_j\in V_2, ~1\leq i \leq |V_1|,1\leq j \leq |V_2|\}$.  Note that $H$ is a bisplit graph with $V_1\cup V_2$ as biclique and $V_3$ as an independent set. Clearly, $H$ is chordal bipartite. Thus, $H$ is a chordal bipartite bisplit with $I$ adjacent to one biclique partition. We claim that $G$ is a yes-instance of HCYCLE if and only if $H$ is a yes-instance of HPATH.

Since the proof of this theorem is similar to the proof presented in Theorem \ref{hcbi}, we present a sketch of this proof. For the forward direction, suppose $G$ has a Hamiltonian $C$. It is clear that $C$ must have exactly one clique edge, say $e=\{x_i,x_j\}$.   Clearly, $C$ without the clique edge $e$ is a path $P_{x_ix_j}$ spanning $V(G)$. We shall now construct the path $P'$ as follows. We use the same order in which the path $P_{x_ix_j}$  visits the vertices of $G$ to visit the vertices of $V_1\cup V_3$, in particular $\{x^1_1,\ldots,x^1_{|K|}\}\cup V_3$. Note that the clique edge $\{x_i,x_j\}$ in $P$ is modified to a $P_3=\{x^1_i,x^2_j,x^1_j\}$.  Since $V_1 \cup V_2$ is a biclique, the leftover vertices of $V_1 \cup V_2$ are visited by alternating $V_1 \cup V_2$. Observe that the endpoints of the constructed path $P'$ are in $V_1$ and $V_2$. Since $V_1 \cup V_2$ is a biclique, we obtain the Hamiltonian cycle $C'$. Conversely, suppose $H$ contains a Hamiltonian cycle $C'$, we construct the corresponding Hamiltonian cycle $C$ in $G$. By our construction, $V_3$ is adjacent to only the first half of $V_1$ and $|V_1|=|V_3|+1$. This shows that in any Hamiltonian cycle, the vertices $\{x^1_1,\ldots,x^1_{|K|}\}\cup V_3$ must form a subpath $P'$. Since $|V_1|=|V_3|+1$, the end points of $P'$ must be in $V_1$.  We know that $V_1$ is a clique in $G$. Since the endpoints of $P'$ are in $K$, $P'$ along with the clique edge a Hamiltonian cycle in $G$. 
\end{proof}
	
	\subsection{Hardness Result on $P_{10}$-free chordal bipartite graphs}
		\begin{theorem} \label{p10cycle}
		For $P_{10}$-free chordal bipartite graphs, the Hamiltonian cycle problem is NP-complete.
	\end{theorem}
	\begin{proof}
		We present a deterministic polynomial-time reduction that reduces an instance of a split graph $G(K,I)$ with $n\geq3$ vertices to the corresponding $P_{10}$-free chordal bipartite instance $H(X,Y)$. We use gadgets to construct $H$. The gadget construction is as follows. 	  For each vertex $u_i$ in $G$, we create a vertex gadget $X_i$ in $H$ as follows:  $V(X_i)=\{x^i_{l} ~|~1\leq l \leq14\}$  and $E(X_i)=\{\{w,z\}~|~w\in\{x^i_2,x^i_4,x^i_6,x^i_{14}\},z\in\{x^i_1,x^i_7,x^i_9,x^i_{11}\}\}\cup \{\{x^i_{8},x^i_{7}\},\{x^i_{8},x^i_{9}\},$ $ \{x^i_{8},x^i_{11}\},\{x^i_{10},x^i_{9}\},\{x^i_{10},x^i_{11}\},\{x^i_{12},x^i_{11}\},\{x^i_{12},x^i_{13}\},\{x^i_{3},x^i_{2}\},\{x^i_{3},x^i_{4}\},\{x^i_{5},x^i_{4}\}, \{x^i_{5},x^i_{6}\},\\\{x^i_{13},x^i_{2}\},\{x^i_{13},x^i_{14}\}\}$. The gadget construction is illustrated in Figure \ref{fig:p10gad}. We shall now define the vertex set of $H$, $V(H)=\bigcup_{i=1}^{n}V(X_i)$. The edge set $E(H)=E'\cup E''$ which is defined as follows,  $E'=\bigcup_{i=1}^{n}E(X_i)$ and  $E''=\{\{w,z\}~|~w\in\{x^i_2,x^i_4,x^i_6,x^i_{14},x^j_2,x^j_4,x^j_6,x^j_{14}\},z\in\{x^i_7,x^i_9,x^i_{11},x^j_7,x^j_9,x^j_{11}\}\}$.	Note that $E''$ defines the edges between $\{x^i_2,x^i_4,x^i_6,x^i_{14},x^j_2,x^j_4,x^j_6,x^j_{14}\}$ and $\{x^i_7,x^i_9,x^i_{11},x^j_7,x^j_9,x^j_{11}\}$,  which form a complete bipartite subgraph in $H$.  An example is illustrated in Figure \ref{fig:p10con}. \\
 		\begin{figure}[hbt] 
			\begin{center}
				\includegraphics[width=70mm,scale=0.5]{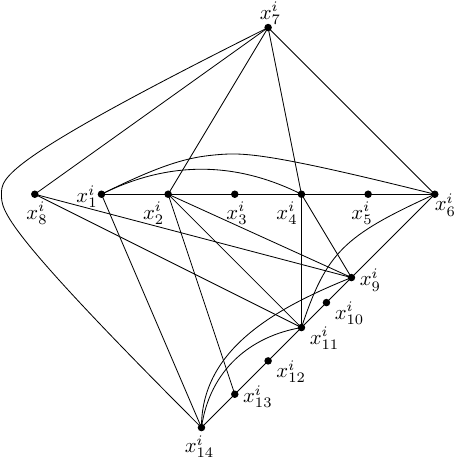}
				\caption{\small \sl   The vertex gadget $X_i$ in $H$ corresponding to the vertex $u_i$ in $G$\label{fig:p10gad}} 
			\end{center}
		\end{figure}\\
		\noindent We shall now show that $H$ is a $P_{10}$-free chordal bipartite graph. We first show that $H$ is bipartite. By the construction, note that any pair of odd-numbered vertices is not adjacent, and so is even-numbered vertices. This shows that the odd-numbered vertices (resp. even-numbered vertices) form an independent set. Therefore, $H$ is bipartite. 
		
		We now show that $H$ is a chordal bipartite graph. Let $Z_1=\{x^j_2,x^j_4,x^j_6,x^j_{14}\}$ or $Z_2=\{x^j_7,x^j_9,x^j_{11}\}$. Let $C_k, k\geq6$ be a cycle of length at least six in $H$. If $C_k$ involves the vertices of $X_i$ alone, then there is a chord in $C_k$ by our construction. Therefore, each gadget $X_i$ is chordal bipartite by the construction. We now consider the following cases to show that there is always a chord in $C_k$. 
		
		\textbf{Case 1:} $C_k$ with the vertices of $X_i$ and $X_j$. Note that $C_k$ must contain at least one vertex from $X_j$, particularly from $Z_1$ or $Z_2$. 
		
		\textbf{Case 1.1:} $C_k$ contains exactly one vertex $w$ from $X_j$. Without loss of generality, assume that $w\in Z_1$. By the construction, $N_G(w)\cap V(X_i)=\{x^i_7,x^i_9,x^i_{11}\}$. In $C_k$, $w$ is adjacent to any two of $\{x^i_7,x^i_9,x^i_{11}\}$. Let $p$ and $q$ be the neighbors of $w$ in $C_k$ ($p,q\in \{x^i_7,x^i_9,x^i_{11}\}$).  It is clear that other vertices (at least five vertices) of $C_k$ must be part of $X_i$. By our construction, any path $P_l$ of length at least five connecting $p$ and $q$ in $X_i$ has either a chord from $w$ to internal vertices of $P_l$ or chord from $p$ to internal vertices of $P_l$ or chord from $q$ to internal vertices of $P_l$. This shows that $C_k$ is a cycle of length at least six with a chord in it.
		
		\textbf{Case 1.2:} $C_k$ contains at least two vertices $w_1$ and $w_2$ from  $X_j$. 
		Case 1.2.1: $w_1$ and $w_2$ are adjacent in $C_k$.  It must be the case that $w_1\in Z_1$ and $w_2\in Z_2$ or vice versa. Observe that, in $C_k$, $w_1$ is adjacent to any one of $\{x^i_7,x^i_9,x^i_{11}\}$, say $p$ and $w_2$ is adjacent to any one of $\{x^i_2,x^i_4,x^i_{6},x^i_{14}\}$, say $q$. Since $p$ and $q$ are adjacent in $X_i$, the edge $\{p,q\}$ is a chord in $C_k$.  This shows that $C_k$ has a chord in it. Case 1.2.2: $w_1$ and $w_2$ are non-adjacent in $C_k$. $C_k$ must contain at least two vertices from $Z_1$.  Without loss of generality, assume that $w_1,w_2\in Z_1$. Observe that, in $C_k$, $w_1$ (resp. $w_2$) is adjacent to any one of $\{x^i_7,x^i_9,x^i_{11}\}$ say $p$ (resp. $q$). The edge $\{w_1,q\}$ (resp. $\{w_2,p\}$) is a chord in $C_k$. 
		
		\textbf{Case 2:} $C_k$ with the vertices of $X_i, X_j$ and $X_l$. Case 2.1: $X_i$ and $X_l$ are adjacent. Case 2.1.1: $C_k$ contains exactly one vertex $w$ from $X_j$. Without loss of generality, we assume that $w\in Z_1$. This shows that $w$ in $C_k$ must be adjacent to one of $\{x^i_7,x^i_9,x^i_{11}\}$, say $p$ and one of $\{x^l_7,x^l_9,x^l_{11}\}$, say $q$.  Suppose that $C_k$ contains only $q$ from $X_l$ then it must be adjacent to one of $\{x^i_2,x^i_4,x^i_{6},x^i_{14}\}$ in $X_i$, say $r$. Clearly, the edge $\{p,r\}$ is a chord in $C_k$. Now we consider the case where $C_k$ contains one more vertex other than $q$ from $X_l$, say $s$. It must be the case that $s$ is one of $\{x^l_2,x^l_4,x^l_{6},x^l_{14}\}$. The edge $\{p,s\}$ in $C_k$ is a chord. Case 2.1.2: $C_k$ contains two vertices $w_1$ and $w_2$ from $X_j$. Without loss of generality, assume that $w_1\in Z_1$. Suppose $C_k$ contains exactly one vertex from $X_l$, then $w_2\in Z_1$. The edge between the neighbor of $w_1$ in $V(C_k)\cap X_i$ and $w_2$ is a chord in $C_k$. Assume that $C_k$ contains two vertices from $X_l$. In this case $w_2\in Z_2$ and $\{w_1,w_2\}$ is a chord in $C_k$.   
		
		\textbf{Case 2.2:} $X_i$ and $X_l$ are non-adjacent. Since $X_i$ and $X_l$ are non-adjacent, there must exist at least two vertices from $X_j$. An argument similar to that of Case 2.1.2 is true for this case as well.
		
		\textbf{Case 3:} $C_k$ with the vertices of at least four gadgets in the following order $X_i, X_j, X_r$ and $X_s$. It is known that $G$ is a split graph and hence $G$ does not contain an induced $C_4$. Therefore, these four gadgets cannot induce a cycle of length at least six in $H$. Now, we shall look at the other possible cases. Case 3.1: $X_i$ and $X_r$ are adjacent. Case 3.2: $X_i$ and $X_r$ are non-adjacent. The proof is similar to that of Case 2.2 and Case 2.2. All the cases clearly show that $C_k$ has a chord in it. Note that the above arguments are true for any $C_k,k\geq 6$ in $H$. Therefore, $H$ is a chordal bipartite graph.
		We now show that $H$ is $P_{10}$-free.
		\begin{claim} \label{p10claim}
			 If $H$ has an induced path $P_r$, $r\geq 10$, then $G$ has an induced path $P_s$, $s\geq 5$.\end{claim}
		\begin{proof}
			 Let $P_r$ denote an induced path of length at least ten in $H$. We say two gadgets ($X_i,X_{i+1}$)  are adjacent if $\{u,v\}\in E(G), u\in X_i$ and $v\in X_{i+1}$. By our construction, we know that each gadget is $P_{10}$-free. This shows that any $P_r$ is such that the internal vertices of $P_r$ spans more than one gadget in $H$. We now claim that $P_r$ must span at least five gadgets in $H$. Let $Q$ be the collection of gadgets that contribute vertices to $P_r$. We now show that $|Q|\geq5$, i.e., $Q=\{X_i,X_{i+1},X_{i+2},X_{i+3},X_{i+4}\}$. We shall now analyze the contribution of $X_{i}$ and $X_{i+1}$ to $P_{r}$. Observe that $X_{i}$ can contribute at most one vertex from $\{x^i_2,x^i_4,x^i_{6},x^i_{14}\}$ (resp. $\{x^i_7,x^i_{9},x^i_{11}\}$) to $P_{r}$. Without loss of generality, assume that $x^i_6$ and $x^i_7$ are part of $P_{r}$. By our construction, $x^i_6$ and $x^i_7$ are adjacent. Thus we obtain an induced path $P=(x^i_6,x^i_7)$.  The path $P$ can be extended to by adding $x^i_5 (x^i_8)$, i.e, $P=(x^i_5,x^i_6,x^i_7)$. Note that $P$ cannot be extended further using the gadget $X_i$. Now the possible vertices from $X_{i+1}$ to extend $P$ are one of $\{x^{i+1}_2,x^{i+1}_4,x^{i+1}_{6},x^{i+1}_{14}\}$. Suppose it uses two vertices from $\{x^{i+1}_2,x^{i+1}_4,x^{i+1}_{6},x^{i+1}_{14}\}$ then by our construction $x^i_7$ will be adjacent to both of them, which forms a chord. Thus $P$ is extended as follows $P=(x^i_5,x^i_6,x^i_7,x^{i+1}_6)$. Note that further extension using the vertices of $X_{i+1}$ is not possible as $x^i_6$ is adjacent to all of $\{x^{i+1}_7,x^{i+1}_{9},x^{i+1}_{11}\}$, and also $x^{i+1}_1$ and $x^{i+1}_5$ have adjacency only within the same gadget. This shows that two gadgets cannot form $P_r$ in $H$.
			
		We now consider the contribution of $X_{i+2}$ to $P_r$. Suppose that the gadgets $X_i$ and $X_{i+2}$ are adjacent, $P=(x^i_5,x^i_6,x^i_7,x^{i+1}_6)$ cannot be extended as $x^i_6$ and $x^{i+1}_6$ are adjacent to all of $\{x^{i+2}_7,x^{i+2}_{9},x^{i+2}_{11}\}$. Therefore, the gadgets $X_i$ and $X_{i+2}$ must be non-adjacent. Since $X_i$ and $X_{i+2}$ are non-adjacent, two vertices from $X_{i+2}$ can be used, in particular one of $\{x^{i+2}_2,x^{i+2}_4,x^{i+2}_{6},x^{i+2}_{14}\}$ and one of $\{x^{i+2}_7,x^{i+2}_9,x^{i+2}_{11}\}$ to extend $P$, $P=(x^i_5,x^i_6,x^i_7,x^{i+1}_6,x^{i+2}_7,x^{i+2}_6)$. Note that $P$ cannot be extended further using the gadget $X_{i+2}$. Next, we consider the gadget $X_{i+3}$ and it must be non-adjacent to $X_i$ and $X_{i+1}$. Suppose $X_{i+3}$ is adjacent to $X_i$ or $X_{i+1}$, by our construction $P$ cannot be extended further as it forms a chord.  To extend the path $P$, the gadget $X_{i+3}$ must be adjacent to $X_{i+2}$.  Since $X_{i+2}$ and $X_{i+3}$ are adjacent and $P$ contains two vertices from $X_{i+2}$, only one vertex from $\{x^{i+3}_7,x^{i+3}_{9},x^{i+3}_{11}\}$ can be added to extend the path $P$, i.e $P=(x^i_5,x^i_6,x^i_7,x^{i+1}_6,x^{i+2}_7,x^{i+2}_6,x^{i+3}_{11})$. By our construction, the vertices $x^{i+3}_{12}$ and $x^{i+3}_{13}$ are used to extend the path $P=(x^i_5,x^i_6,x^i_7,x^{i+1}_6,x^{i+2}_7,x^{i+2}_6,x^{i+3}_{11},x^{i+3}_{12},x^{i+3}_{13})$. Note that $P$ cannot be extended further as the neighbors of $x^{i+3}_{13}$ forms a chord. Observe that the  length of $P$ is nine. In order to extend $P$ further, we must use the gadget $X_{i+4}$ and it must be non-adjacent to $X_i,X_{i+1}$ and $X_{i+2}$. Thus $|Q|\geq 5$. Clearly, the vertices correspond to the gadgets in $Q$ induce a $P_s$, $s\geq5$ in $G$. This implies that $G$ has an induced $P_s$, $s\geq5$ in $G$. \end{proof}	
		By Claim \ref{p10claim}, it is clear that if $H$ has an induced $P_{r}$, $r\geq 10$ then $G$ has an induced $P_s$, $s\geq5$. Note that $G$ (split graphs) is $P_5$-free. Therefore, $H$ is $P_{10}$-free.

		We claim that $G$ is a yes-instance of the Hamiltonian cycle problem if and only if $H$ is a yes-instance of the Hamiltonian cycle problem.

		\begin{figure} [!h]
			\begin{center}
				\includegraphics[width=130mm,scale=0.5]{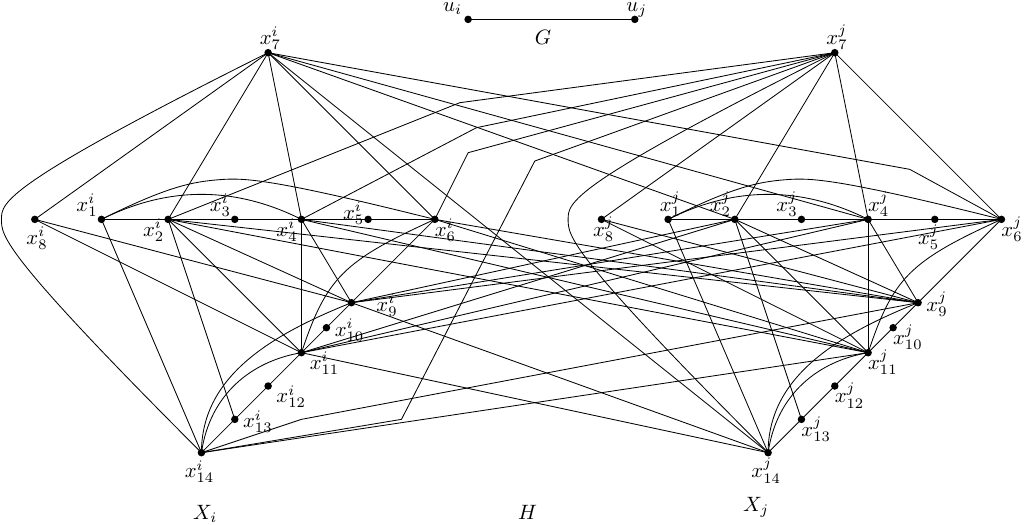}
				\caption{\small \sl   The edge gadget corresponding to the edge $\{u_i,u_j\}$ in $G$ \label{fig:p10con}} 
			\end{center}
		\end{figure}
		\emph{Necessity:} Assume that a Hamiltonian cycle $C$ exists in $G$. We shall now construct the corresponding Hamiltonian cycle $C'$ in $H$. For each gadget $X_i$ in $H$, there exists a path $P_{x^i_7x^i_6}$ that spans the vertices of $X_i$ in the following order, $P_{x^i_7x^i_6}=(x^i_7,x^i_{8},x^i_9,x^i_{10},x^i_{11},x^i_{12},x^i_{13},x^i_{14},x^i_1,x^i_2,x^i_3,x^i_4,x^i_5,x^i_6)$. For each edge $\{u_i,u_{i+1}\}$ in $C$, we make use of $P_{x^i_7x^i_6}$, the edge $\{x^i_6,x^{i+1}_7\}$, and $P_{x^{i+1}_7x^{i+1}_6}$ to construct $C'$ in $H$. For example, $C=(\ldots,u_i,u_{i+1},\ldots)$ in $G$ is transformed to $C'=(\ldots,P_{x^i_7x^i_6},P_{x^{i+1}_7x^{i+1}_6},\ldots)$ in $H$. The transformation of each edge of $C$ in $G$ gives us the desired Hamiltonian cycle $C'$ in $H$. 
		
		\emph{Sufficiency:}
		Suppose a Hamiltonian cycle $C'$ exists in $H$. We now argue that for each gadget $X_i$, $C'$ must visit all of $X_i$ in one pass before it goes to the next gadget, and there is no partial visit for any gadget. Note that for any  gadget $X_i$, the degree of the following vertices $\{x^i_3,x^i_5, x^i_{10}$, $x^i_{12}\}$ are two. Therefore, the subpaths $P_{x^i_2x^i_6}=(x^i_2,x^i_3,x^i_4,x^i_5,x^i_6)$  and $P_{x^i_9x^i_{13}}=(x^i_9,x^i_{10},x^i_{11},x^i_{12},x^i_{13})$ must appear in order in any Hamiltonian cycle, in $C'$ as well.  We now analyze the possible appearance of $x^i_{1},x^i_{7},x^i_{8}$ and $x^i_{14}$ in $C'$. Observe that $x^i_{8}$ cannot be adjacent to $x^i_{11}$ in $C'$. Therefore, $x^i_{8}$  must be adjacent to $x^i_7$ and $x^i_9$. Thus the path $P_{x^i_9x^i_{13}}$ is extended to; $P_{x^i_7x^i_{13}}=(x^i_7,x^i_{8},x^i_9,x^i_{10},x^i_{11},x^i_{12},x^i_{13})$. On the similar line, the vertex $x^i_1$ must be adjacent to any two of $\{x^i_{2},x^i_{6},x^i_{14}\}$ in $C'$. Suppose that $x^i_{1}$ is adjacent to both $x^i_{2}$ and $x^i_{6}$, then $P_{x^i_2x^i_6}$ cannot be extended further. Thus, we have the following cases.
		
		\textbf{Case 1:}  $x^i_{1}$ is adjacent to $x^i_2$ and $x^i_{14}$ in $C'$. The path $P_{x^i_2x^i_6}$ is extended to $P_{x^i_{14}x^i_6}=(x^i_{14},x^i_1,x^i_2,x^i_3,x^i_4,x^i_5,x^i_6)$.  It is now clear that $x^i_{13}$ must be adjacent to $x^i_{14}$ in $C'$. Thus we obtain the spanning path $P_{x^i_7x^i_{6}}$ spanning $X_i$ in the following order $P_{x^i_7x^i_{6}}=(x^i_7,x^i_{8},x^i_9,x^i_{10},x^i_{11},x^i_{12},x^i_{13},x^i_{14},x^i_1,x^i_2,x^i_3,x^i_4,x^i_5,x^i_6)$. This shows that the entry to the gadget is through $x^i_7$ (resp. $x^i_6)$  and exit is through $x^i_6$ (resp. $x^i_7)$.
		
		\textbf{Case 2:}  $x^i_{1}$ is adjacent to $x^i_6$ and $x^i_{14}$ in $C'$. In this case, we obtain the following two spanning paths spanning $X_i$:
		
		$P_{x^i_7x^i_{2}}=(x^i_7,x^i_{8},x^i_9,x^i_{10},x^i_{11},x^i_{12},x^i_{13},x^i_{14}, x^i_{1},x^i_6,x^i_5,x^i_4,x^i_3,x^i_2)$ and
		
		$P_{x^i_7x^i_{14}}=(x^i_7,x^i_{8},x^i_9,x^i_{10},x^i_{11},x^i_{12},x^i_{13},x^i_2,x^i_3,x^i_4,x^i_5,x^i_6,x^i_1,x^i_{14})$. 
		
		This shows that $C'$ enters the gadget $X_i$ through $x^i_7$ and exit through one of $x^i_{2}$ or $x^i_{14}$.  		It is clear from both cases that each gadget must be visited completely before the cycle visits the next gadget. Thus, vertices in $X_i$ are visited in one pass, followed by the vertices in $X_{i+1}$. To obtain the desired Hamiltonian cycle in $G$, we replace the gadget $X_i$ with the corresponding vertex $u_i$ in $G$.

	\end{proof}
\noindent \textbf{Remarks:} Theorem \ref{p10cycle} shows that for $P_{10}$-free chordal bipartite graphs with $n\geq 42$, the Hamiltonian cycle problem is NP-complete. For $n<42$, we can trivially solve the Hamiltonian cycle problem by exploring all possible solutions.

Using the above construction with the Hamiltonian path in split graphs as a candidate problem, we can show that the Hamiltonian path problem is NP-complete on $P_{10}$-free chordal bipartite graphs.
	\subsection{Polynomial results on $P_5$-free chordal bipartite graphs}\label{p5}
	In this section, we present a linear-time algorithm for HCYCLE (HPATH), which acts as a framework to solve a sequence of problems that are variants and generalizations of HCYCLE (HPATH). To obtain a solution to these variants and generalizations problems, we carry out a small computation, which reduces these problems to an instance of HCYCLE (HPATH). Further, we use  HCYCLE (HPATH) solution to solve these variants and generalizations. 
	
	In \cite{bipartitechain}, it is shown that $P_5$-free bipartite graphs are bipartite chain graphs. This shows that $P_5$-free chordal bipartite bisplit graphs are bipartite chain graphs. It is known that HCYCLE \cite{hcconvex} and HPATH \cite{BChpath} are linear-time solvable in bipartite chain graphs. However, the algorithms available in the literature for $P_5$-free chordal bipartite bisplit graphs cannot be used as a framework to solve the variants and generalizations of HCYCLE (HPATH) that are reported in this paper. With this motivation, we revisit and discover a new approach to solve HCYCLE (HPATH), using which we obtain linear-time algorithms for variants and generalizations in $P_5$-free chordal bipartite bisplit graphs. 
	
	\subsubsection{Structural Results on $P_{5}$-free Chordal Bipartite Graphs}
	\label{stresults}
	For a chordal bipartite graph $G$ with bipartition $(A,B)$, let $A=\{x_1,x_2,\ldots,x_m\}$ and $B=\{y_1,y_2,\ldots,y_n\}$.   For a $P_{5}$-free chordal bipartite graph, we observe that $A=A_{1} \cup A_{2}$ and $B=B_{1} \cup B_{2}$, $A_{1}=\{x_{1},x_{2},\ldots,x_{i}\}$, $A_{2}=\{x_{i+1},x_{i+2},\ldots,x_{m}\}$ and $B_{1}=\{y_1,y_2,\ldots,y_j\}$, $B_{2}=\{y_{j+1},y_{j+2},\ldots,y_{n}\}$, such that $(A_1,B_1)$ is a biclique and $A_2$ and $B_2$ are independent sets.  The base biclique $(A_1,B_1)$ is a maximal clique $K_{i,j}$ such that $|i-j|$ is minimum over all such maximal bicliques. The procedure to find the base biclique $(A_1,B_1)$ is as follows. First, we arrange the vertices of $A$ in non-increasing order. Let $x_1$ be the first vertex. Observe that $x_1$ together with its neighbours induces a biclique $K_{1,d_G(x_1)}$. Similarly, fix the first two vertices and their neighbors to find a biclique $K_{2,d_G(x_2)}$. By following this procedure, we find the base biclique $K_{i,j}$ such that $|i-j|$ is minimum.\\ 
	\begin{lemma}
		\label{lem1}
		Let $G$ be a $P_5$-free chordal bipartite graph. Then, $\forall x \in A_{2},\exists y_{k} \in B_{1}$ such that $y_{k} \not\in N_G(x)$ and $\forall y \in B_{2},\exists x_{k} \in A_{1}$ such that $x_{k} \not\in N_G(y)$.
	\end{lemma}
	\begin{proof}
		Suppose there exists $x$ in $A_2$ such that for all $y_k$ in $B_1$, $xy_k \in E(G)$.  Then $(A_1 \cup \{x\},B_1)$ is the base biclique, contradicting the fact that $(A_1,B_1)$ is maximum.  A similar argument is true for $y \in B_2$. Therefore, the lemma follows.  
	\end{proof}
	\begin{lemma}
		\label{lem2}
		Let $G$ be a $P_5$-free chordal bipartite graph. Then,  $\forall x \in A_{2}$, $N_G(x) \subset B_{1}$ and $ \forall y \in B_{2}$, $N_G(y) \subset A_{1}$. 
	\end{lemma}
	\begin{proof}
		On the contrary,  $\exists x_{a} \in A_{2}$, $N(x_{a}) \not\subset B_{1}$. Case 1: $N(x_{a}){=}B_{1}$. Then, $(A_{1} \cup \{x_{a}\}, B_{1})$ is the  base biclique, a contradiction.  Case 2: $N(x_{a}) \cap B_{2} \neq \emptyset$.  This implies that there exists $y_{b} \in B_2$ such that $y_b \in N(x_{a})$.  Since $G$ is connected, $N(y_{b}) \cap A_{1} \neq \emptyset$, say $x_{c} \in N(y_{b}) \cap A_{1}$.  In $G$, $P(x_{a},x_{k})=(x_{a},y_{b},x_{c},y_{k},x_{k})$ is an induced $P_{5}$.  Note that, due to the maximality of $(A_1,B_1)$, as per Lemma \ref{lem1}, we find $x_{k} \not\in N(y_{b}), y_{k} \not\in N(x_{a})$.  This contradicts that $G$ is $P_5$-free.  Case 3: $N(x_{a}) \cap B_{1} \neq \emptyset$ and $N(x_{a}) \cap B_{2} \neq \emptyset$. Observe that $ \exists x_{a} \in A_{2}$ such that $y_{b},y_{c} \in N(x_{a})$ and $y_{b} \in B_{1}$, $y_{c} \in B_{2}$.  In $G$, $P(y_{c},y_{k})=(y_{c},x_{a},y_{b},x_{k},y_{k}$), $x_{k} \not\in N(y_{c})$, $y_{k} \not\in N(x_{a})$ is an induced $P_{5}$.  Note that the existence of $x_k,y_k$ is due to Lemma \ref{lem1}. This contradicts that $G$ is $P_5$-free. Similarly, $ \forall y \in B_{2}$, $N(y) \subset A_{1}$, can be proved. 
	\end{proof}
	\begin{lemma}
		\label{lem3}
		Let $G$ be a $P_5$-free chordal bipartite graph.  
		For $x_i,x_j \in A_{2}$, $i \not= j$, if $d(x_i) \leq d(x_j)$, then $N(x_i) \subseteq N(x_j)$.   Similarly, for $y_i,y_j \in B_{2}$, if $d(y_i) \leq d(y_j)$, $N(y_i) \subseteq N(y_j)$.
	\end{lemma}
	\begin{proof}
		Let us assume to the contrary that $N(x_i) \not\subseteq N(x_j)$, i.e., $N(x_i)  \setminus  N(x_j) \neq \emptyset$. 
		
		\textbf{Case 1:} $N(x_i) \cap N(x_j)\neq \emptyset$. This implies that $\exists y_{a}\in B$ such that  $y_{a} \in N(x_i) \cap N(x_j)$.  Since $N(x_i) \not\subseteq N(x_j)$, $\exists y_{b}\in B$ such that  $y_{b} \not\in N(x_i) \cap N(x_j)$ and $y_{b} \in N(x_i)$.  Since $d(x_j) \geq d(x_i)$, vertex $x_j$ is adjacent to at least one more vertex $y_{c} \in B$ such that $y_{c} \not\in N(x_i)$.  The path $P(y_{c},y_{b})=(y_{c},x_j,y_{a},x_i,y_{b}$) is an induced $P_{5}$.  This is a contradiction. 
		
		 \textbf{Case 2:} $N(x_i) \cap N(x_j)= \emptyset $, $ | N(x_i) | \geq 1$ and $ | N(x_j) | \geq 1$. We observe that  $\exists y_{a},y_{b}\in B$ such that  $y_{a} \in N(x_i)$, $y_{a} \not\in N(x_j)$ and $y_{b} \in N(x_j)$, $y_{b} \not\in N(x_i)$.  Since $G$ is a connected graph,  $|P(y_{a},y_{b})| \geq 3$ is an induced path with at least three vertices. The path $P(x_i,x_j)=(x_i,P(y_{a},y_{b}),x_j$) is an induced path with at least five vertices. This is a contradiction. Similarly, for all pairs of distinct vertices $y_i,y_j \in B_{2}$ with $d(y_i) \leq d(y_j)$, $N(y_i) \subseteq N(y_j)$ can be proved.   
	\end{proof}
	\begin{theorem}
		\label{thm1}
		Let $G$ be a $P_5$-free chordal bipartite graph with $(A_1,B_1)$ being the base biclique.  Let $A_{2}=(u_{1},u_{2},...,u_{p})$ and $B_{2}=(v_{1},v_{2},...,v_{q})$ are orderings of vertices.  If $d_{G}(u_1) \leq d_{G}(u_2) \leq d_{G}(u_3) \leq \ldots \leq d_{G}(u_p)$, then $N(u_{1}) \subseteq N(u_{2}) \subseteq N(u_{3}) \subseteq \ldots \subseteq N(u_{p})$.  Further, if $d_{G}(v_1) \leq d_{G}(v_2) \leq d_{G}(v_3) \leq \ldots \leq d_{G}(v_q)$, then $N(v_{1}) \subseteq N(v_{2}) \subseteq N(v_{3}) \subseteq \ldots \subseteq N(v_{q})$.
	\end{theorem} 
	\begin{proof} 
		We shall prove by mathematical induction on $|A_2|$.  Base Case: $|A_{2}|=2, A_2=(u_{1},u_{2})$ such that $d(u_{1}) \leq d(u_2)$.  By Lemma \ref{lem3}, $N(u_1) \subseteq N(u_2)$.  Induction step: Consider $A_2=(u_{1},u_{2},u_{3},\ldots,u_{p}), p \geq 3$.  Consider the vertex $u_{p} \in A_{2}$ such that $d(u_p) \geq d(u_{p-1})$.  By Lemma \ref{lem3}, $N(u_{p-1}) \subseteq N(u_{p})$ is true.  By the hypothesis, $N(u_{1}) \subseteq N(u_{2}) \subseteq N(u_{3}) \subseteq \ldots \subseteq N(u_{p-1})$.  By combining the hypothesis and the fact that $N(u_{p-1}) \subseteq N(u_{p})$, our claim follows.  Similarly for $B_{2}$ as well.
	\end{proof}
\noindent	We refer to the above ordering of vertices as {\em Nested Neighbourhood Ordering (NNO)} of $G$.  From now on, we shall arrange the vertices in $A_2$ in the non-decreasing order of their degrees to work with {\em NNO} of $G$.    
	\subsection{Hamiltonicity in $P_{5}$-free chordal bipartite graphs}
	\label{hamil} 
	For a connected graph $G$ and the set $S \subset V(G)$, $c(G-S)$ denotes the number of connected components in the graph induced on the set $V(G) \setminus S$. It is well-known, due to Chv{\'a}tal \cite{dbwest2003} that if a graph $G$ has a Hamiltonian cycle, then for every $S\subset V(G), c(G-S) \leq |S|$. Similarly, if a graph $G$ has a Hamiltonian path, then for every $S\subset V(G), c(G-S) \leq |S|+1$. \\
	\begin{theorem}
		\label{thm2}
		For a $P_5$-free chordal bipartite graph $G$, $G$ has a  Hamiltonian cycle if and only if (i) $|A| = |B|$ and (ii) $A_{2}$ has an ordering $(u_{1},u_{2},\ldots,u_{p})$, such that $\forall u_g, d(u_g) > g$, $1 \leq g \leq p$ and $B_{2}$ has an ordering $(v_{1},v_{2},\ldots,v_{q})$, $ \forall v_h, d(v_h) >h$, $1 \leq h \leq q$.
	\end{theorem}
	\begin{proof}
		{\em Necessity:} (i) Any cycle in a bipartite graph is even and alternates between vertices of $A$ and $B$. Since the Hamilton cycle visits all the vertices in $A$ and $B$, it must have $ |A|= |B|$. (ii) On the contrary, $ \exists u_g \in A_{2}$  such that $u_g$ is the first vertex in the ordering with $d(u_g) \leq g$.  That is, for $u_k \in \{u_1,\ldots,u_{g-1}\}$, $d_{G}(u_k) > k$ and $d(u_g) \leq g$.  Since $G$ follows {\em NNO}, $d_{G}(u_g)= g$.  From Theorem \ref{thm1},  we know that $N(u_1) \subseteq N(u_2) \subseteq \ldots \subseteq N(u_{g-1}) \subseteq N(u_g)$.   This implies that $c(G-N(u_g)) = g+1 >g$.  This is a contradiction to Chv{\'a}tal's necessary condition for the Hamiltonian cycle. Similarly, in $B_{2}$, $ \forall v_h,$ $d(v_h)>h$ can be proved. 
		
		{\em Sufficiency:} Let $i=|A_1|$ and $j=|B_1|$. Since $A_{2}$ has an ordering such that $ \forall u_g \in A_2$, $d(u_g)>g$, for clarity purpose, we define $N_G(u_g)$ as follows; $N(u_g)=\{y_{1},y_{2},\ldots,y_{l}\}$, $g<l<j$,  that is, $u_1$ is adjacent to at least two vertices $\{y_{1},y_{2}\}$ and at most $j-1$ vertices $\{y_{1},y_{2},\ldots,y_{j-1}\}$, $u_2$ is adjacent to at least three vertices $\{y_{1},y_{2},y_{3}\}$ and at most $j-1$ vertices $\{y_{1},y_{2},\ldots,y_{j-1}\}$ and similarly $u_p$ is adjacent to at least $p+1$ vertices $\{y_{1},y_{2},\ldots,y_{p+1}\}$ and at most $j-1$ vertices $\{y_{1},y_{2},\ldots,y_{j-1}\}$.  Observe that, due to the maximality of $(A_1,B_1)$, any $u_g$ of $A_2$ can be adjacent to at most $j-1$ vertices of $B_1$. Similarly, in $B_2$, for all $v_h \in B_2$,               $N(v_h)=\{x_{1},x_{2},\ldots,x_{l}\},$ $h<l<i$. 
		
		Let $d(u_p)=r, p+1 \leq r \leq j-1$ and $d(v_q)=s, q+1 \leq s \leq i-1$. The vertices in $A_{1}$ can be ordered as $(x_{1},x_{2},\ldots,x_{q},x_{q+1},\ldots,x_{i})$ and the vertices in $B_{1}$ can be ordered as $(y_{1},y_{2},\ldots,y_{p},y_{p+1},\ldots,y_{j})$.  
		Note that $A_{3}=A_{1}{\setminus}\{x_{1},x_{2},\ldots,x_{q+1}\}=\{x_{q+2},\ldots,x_{i-1},x_{i}\}$ and $B_{3}=B_{1}{\setminus}\{y_{1},y_{2},\ldots,y_{p+1}\}=\{y_{p+2},\ldots,y_{j-1},y_{j}\}$. 
		Further, $|A_{3}| = | A | -(| A_{2} | +q+1)= | A | -(p+q+1)$ and $ | B_{3} | = | B | -( | B_{2} | +p+1)= | B | -(q+p+1)$. 
		Since $ | A | = | B | $, it follows that $ | A_{3} | = | B_{3} | $. 
		In $G$,  $(y_{1},u_{1},y_{2},u_{2},\ldots,y_{p},u_{p},y_{p+1},x_{1},v_{1},x_{2},v_{2},\ldots,x_{q},v_{q},x_{q+1},y_{p+2},x_{q+2},$ $\ldots,y_{j},x_{i},y_{1})$ is a Hamiltonian cycle. 
	\end{proof}
	\begin{theorem}
		\label{thm3}
		For a $P_5$-free chordal bipartite graph $G$, $G$ has a  Hamiltonian path  if and only if one of the following is true \\
		(i)  $|A|=|B|$ and $A_2$ has an ordering,  $\forall u_g, d(u_g) \geq g$, $1 \leq g \leq p$ and $B_{2}$ has an ordering, $ \forall v_h, d(v_h) \geq h$, $1 \leq h \leq q$. \\
		(ii) $|A|=|B|+1$ and $A_2$ has an ordering,$\forall u_g, d(u_g) \geq g$, $1 \leq g \leq p$ and $B_{2}$ has an ordering, $ \forall v_h, d(v_h) >h$, $1 \leq h \leq q$. 
	\end{theorem}
	\begin{proof}
		{\em Necessity:} (i) Without loss of generality, we assume $|A| \geq |B|$.  Any Hamiltonian path starting at $A$ and alternating between $A$ and $B$ can end at $A$ or $B$. Therefore $|A|=|B|$ or $|A|=|B|+1$. To prove that $A_2$ satisfies the ordering, we assume to the contrary that $ \exists u_g \in A_{2}$ such that $u_g$ is the first vertex in the ordering such that $d(u_g)<g$. Since $G$ follows {\em NNO}, $d_{G}(u_g)= g - 1$. From Theorem \ref{thm1}, we know that $N(u_1) \subseteq N(u_2) \subseteq \ldots \subseteq N(u_{g-1}) \subseteq N(u_g)$.  Note that, as per the ordering of $A_2$, $N(u_{g-1})=N(u_g)$. On removing $N(u_g)$ from $G$ we have $g$ components in $A_2$ and $A_1 \cup B_2 \cup (B_1 -N(u_g))$ forms another component.  This implies that $c(G-N(u_g)) =  g+1$.  Clearly, $g+1 \not \leq g-1$.  Thus, we contradict the Chv{\'a}tal's necessary condition for the Hamiltonian path. Similarly $B_{2}$ has an ordering such that $ \forall v_h, d(v_h) \geq h$, $1 \leq h \leq q$. 
		
		(ii) For $A_2$, the argument is similar to the above.   Suppose ${\exists}v_{r}{\in}B_{2}$ such that $v_{r}$ is the first vertex in the ordering such that $d(v_{r}){\leq}r$. From Theorem \ref{thm1}, $N(v_{1}){\subseteq}N(v_{2}){\subseteq}\ldots{\subseteq}N(v_{r-1}){\subseteq}N(v_{r})$. Consider the set $S=B_{1}{\cup}\{v_{r+1},v_{r+2},\ldots,v_{q}\}$ and $|S|=j+q-r-1+1=j+q-r$. Note that the removal of $S$ disconnects $G$. We shall now count the connected components in $c(G-S)$. Firstly, we count the connected components obtained through $A_2$. Since $A_2$ is an independent set adjacent only to $B_1$, we get $p$ connected components. Secondly, the vertices $\{v_{1},v_{2},\ldots,v_{r}\}$ together with  $\{x_{1},x_{2},\ldots,x_{r}\}$ forms one connected component. Thirdly, the remaining vertices of $A_1$ form $i-r$ connected components. Hence, $c(G-S) {\geq} p+1+i-r$.  Note that $|A|=i+p$ and $|B|=j+q$.  Since $|A|=|B|+1$, $c(G-S) {\geq} p+1+i-r = j+q+1+1-r=j+q-r+2$.  Clearly, $c(G-S) \not \leq |S|+1$, contradicting the Chv{\'a}tal's condition for the Hamiltonian path. 
		
		{\em Sufficiency:} (i) Let $N(u_g)=\{y_{1},\ldots,y_{l}\}$, $g \leq l < j$ and $N(v_h)=\{x_{1},\ldots,x_{l}\},$ $h{\leq}l<i$. Consider $A_{3}=A_{1}{\setminus}\{x_{1},x_{2},\ldots,x_{q}\}=\{x_{q+1},x_{q+2},\ldots,x_{i-1},x_{i}\}$.  $B_{3}=B_{1}{\setminus}\{y_{1},y_{2},\ldots,y_{p}\}=\{y_{p+1},x_{p+2},\ldots,y_{j-1},y_{j}\}$.   Note that $|A_3|=|A|-(p+q)$ and $|B_3|=|B|-(p+q)$.  In $G$, \\$P(u_{1},y_{1},u_{2},y_{2},\ldots,u_{p},y_{p},x_{q+1},y_{p+1},x_{q+2},y_{p+2},\ldots,x_{i},y_{j},x_{q},v_{q},\ldots,x_{1},v_{1})$ is a Hamiltonian path.  
		
		(ii) Consider $A_{3}=A_{1}{\setminus}\{x_{1},x_{2},\ldots,x_{q+1}\}=\{x_{q+2},x_{q+3},\ldots,x_{i-1},x_{i}\}$ and 
		$B_{3}=B_{1}{\setminus}\{y_{1},y_{2},\ldots,y_{p}\}=\{y_{p+1},x_{p+2},\ldots,y_{j-1},y_{j}\}$.
		In $G$, $P(u_{1},y_{1},u_{2},y_{2},\ldots,u_{p},y_{p},x_{q+2},y_{p+1},x_{q+3},y_{p+2},\ldots,x_{i},y_{j},x_{q+1},v_{q},\ldots,x_{2},$ $v_{1},x_{1})$ is a Hamiltonian path.  This completes the proof of this claim. 
	\end{proof}
\noindent	As the proofs of Theorems \ref{thm2} and \ref{thm3} are constructive, the Hamiltonian cycle (Hamiltonian path) for $P_5$-free chordal bipartite graphs can be computed in $O(n+m)$ time.  One of the important contributions of this paper is to make use of the results on Hamiltonicity as a framework to solve its variants and generations, which we shall present next.  We also establish that Chv{\'a}tal's necessary condition is sufficient for $P_5$-free chordal bipartite graphs.
	\subsection{Chv{\'a}tal's Necessary condition is Sufficient on $P_5$-free Chordal Bipartite Graphs}
	\label{chv}
	For the Hamiltonian cycle (path) problem, it is well-known from the literature that there are no necessary and sufficient conditions for general graphs. There are a few graphs for which Chv{\'a}tal necessary condition is sufficient, which includes $P_5$-free chordal bipartite graphs. \\
	\begin{theorem}
		Let $G$ be a $P_5$-free chordal bipartite graph with $|A| = |B|$. If $G$ satisfies $c(G-S)\leq|S|$, for every non-empty subset $S$ $\subseteq V(G)$, then $G$ has a Hamiltonian cycle.  
	\end{theorem}
	\begin{proof}
		On the contrary, assume that $G$ has no Hamiltonian cycle, then there exists $u_g$ or $v_h$ in $A_2 (B_2)$ that violates the degree conditions mentioned in Theorem \ref{thm2}.  Let $u_g$ be the first vertex in the ordering with $d(u_g)\leq g$. By Theorem \ref{thm1}, we know that $G$ follows {\em NNO}, $d_{G}(u_g)= g$.  This implies that $c(G-N(u_g)) = g+1 >g$, which is a contradiction to the premise of the theorem. Similarly, $v_h$ in $B_2$ can be proved. 
	\end{proof}
	\begin{theorem}
		Let $G$ be a $P_5$-free chordal bipartite graph with $|A| = |B|$ or $|A| = |B|+1$. If $G$ satisfies $c(G-S)\leq|S|+1$, for every non-empty subset $S$ $\subseteq V(G)$, then $G$ has a Hamiltonian path.  
	\end{theorem}
	\begin{proof}
		\textbf{Case 1: $|A| = |B|$}. Assume on the contrary that $G$ has no Hamiltonian path, then there exists  $u_g$ or $v_h$ in $A_2(B_2)$ that violates degree conditions $(i)$ mentioned in Theorem \ref{thm3}. Let $u_g$ be the first vertex in the ordering with $d(u_g) < g$.  By Theorem \ref{thm1}, we know that $G$ has {\em NNO} and hence $d_{G}(u_g)= g - 1$.  On removing $N(u_g)$ from $G$, we have $g$ components in $A_2$ and $A_1 \cup B_2 \cup (B_1 -N(u_g))$ forms another component.  This implies that $c(G-N(u_g)) =  g+1$.  Clearly, $g+1 \not \leq g-1$, a contradiction.
		
	\textbf{Case 2: $|A| = |B|+1$}. Assume, on the contrary, that $G$ has no Hamiltonian path. Let $v_{h}$ be the first vertex in $B_2$ in the ordering such that $d(v_{h}){\leq}h$. By Theorem \ref{thm1}, $G$ satisfies {\em NNO}.  Consider the set $S=B_{1}{\cup}\{v_{h+1},v_{h+2},\ldots,v_{q}\}$ and $|S|=j+q-h-1+1=j+q-h$.  Further, $c(G-S) {\geq} p+1+i-h$.  Note that $|A|=i+p$ and $|B|=j+q$.  Since $|A|=|B|+1$, $c(G-S) {\geq} p+1+i-h = j+q+1+1-h=j+q-h+2$, contradicting the premise. 
	\end{proof}
	\subsection{Hamiltonicity variants in $P_5$-free Chordal Bipartite Graphs}
	\label{hamilvary}
	In this section, we shall study the computational complexity of some of the variants of Hamiltonicity, such as bipancyclic, homogeneously traceable, EXACTLY 2 SIMPLE PATH COVER,  and hypo-Hamiltonian problems. Interestingly, for all of them,  similar to HCYCLE (HPATH), we do not know the necessary and sufficient conditions in general graphs. In this paper, we establish necessary and sufficient conditions for a few of the variants restricted to $P_5$-free chordal bipartite graphs. To the best of our knowledge, this is the first necessary and sufficient condition for these variants in $P_5$-free chordal bipartite graphs. \\ \\
	{\bf 1. Bipancyclicity} \\
	Pancyclic graphs contain cycles of all possible lengths from three up to the number of vertices in the graph. Since the graph under consideration is bipartite, we study the analogous version of pancyclic graphs, which are bipancyclic graphs.\\
\begin{theorem}
	For a $P_5$-free chordal bipartite graph $G$ of order $2n$, $G$ is Hamiltonian if and only if $G$ is bipancyclic.
\end{theorem}
\begin{proof}
	{\em Necessity:} As $G$ is Hamiltonian,  let $C=(y_{1},u_{1},y_{2},u_{2},\ldots,y_{p},u_{p},y_{p+1},x_{1},v_{1},\\x_{2},v_{2},\ldots,x_{q},v_{q},x_{q+1},y_{p+2},x_{q+2},$ $\ldots,y_{j},x_{i},y_{1})$ be the Hamiltonian cycle obtained from Theorem \ref{thm2}. Clearly, $C$ without the edge $\{x_i,y_1\}$ is a Hamiltonian path $P$ in $G$, $P= (y_{1},u_{1},y_{2},u_{2},\ldots,y_{p},u_{p},y_{p+1},x_{1},v_{1},\\x_{2},v_{2},\ldots,x_{q},v_{q},x_{q+1},y_{p+2},x_{q+2},$ $\ldots,y_{j},x_{i})$. To construct a cycle $C_{2k}$, $2\leq k\leq p$, and $p+2\leq k\leq n$, we consider the subpath $P_{2k}$ on first $2k$ vertices from  $P$. Recall that $y_1$ is universal to $A=A_1\cup A_2$. The construction of $C_{2k}$, $2\leq k\leq p$ and $p+2\leq k\leq n$  is as follows.\\
	$k=2,~C_4=(y_1,u_1,y_2,u_2,y_1)$\\		
	$k=3,~C_6=(y_1,u_1,y_2,u_2,y_3,u_3,y_1)$\\
	$k=3,~C_8=(y_1,u_1,y_2,u_2,y_3,u_3,y_4,u_4,y_1)$\\
	\vdots\\
	$k=p,~C_{2p}=(y_1,u_1,y_2,u_2,\ldots, y_{p-1},u_{p-1},y_p,u_p,y_1)$\\
	$k=p+2,~C_{2(p+1)}=(y_1,u_1,y_2,u_2,\ldots, y_{p-1},u_{p-1},y_p,u_p,y_{p+1},x_1,v_1,x_2,y_1)$\\
	$k=p+3,~C_{2(p+1)}=(y_1,u_1,y_2,u_2,\ldots, y_{p-1},u_{p-1},y_p,u_p,y_{p+1},x_1,v_1,x_2,v_2,x_3,y_1)$\\
	\vdots\\
	$k=n,~C_{2n}=(y_{1},u_{1},y_{2},u_{2},\ldots,y_{p},u_{p},y_{p+1},x_{1},v_{1},x_{2},v_{2},\ldots,x_{q},v_{q},x_{q+1},y_{p+2},x_{q+2},\\\ldots,y_{j-1},x_{i-1},y_{j},x_{i},y_{1})$\\
	Finally, we construct $k=p+1,C_{2(p+1}=(y_1,u_1,y_2,u_2,\ldots, y_{p-1},u_{p-1},y_p,u_p,y_{p+1},\\x_1,y_1)$
	
	Thus, we obtain all even cycles,  $C_{2k}$, $2\leq k\leq n$ in $G$. Hence $G$ is a bipancyclic graph. For the sufficient part, since $G$ is bipancyclic, $G$ contains all possible cycles of length $2k$, $2\leq k \leq n$, including  $C_{2n}$, and therefore, $G$ is Hamiltonian.  Note that, as the proof is constructive, we obtain all cycles in linear time.   
	\end{proof}
	\noindent
	{\bf 2. Homogeneously traceable} 
	\begin{theorem}
		\label{thmhomogeneously}
		For a $P_5$-free chordal bipartite graph $G(A,B)$, $G$ is Hamiltonian if and only if $G$ is homogeneously traceable.
	\end{theorem}
	\begin{proof}
		As the Hamiltonian graphs are homogeneously traceable, the necessity follows. To prove the converse, we observe that, as $G$ is homogeneously traceable, it must be the case that $|A|=|B|$. On the contrary, assume that $|A|\not=|B|$ (say, $|A|=|B|+1$). It is easy to see that a Hamiltonian path does not exist starting at any vertex in $B$, which is a contradiction to the fact that $G$ is homogeneously traceable. Therefore, $|A|=|B|$.  By the definition of homogeneously traceable, we know that $G$ has a Hamiltonian path beginning at each vertex of $G$. We now consider two Hamiltonian paths: the one that starts at $y_1\in B_1$ and the other that starts at $x_1 \in A_1$. Since $G$ has a Hamiltonian path that starts at $y_1$, the degree of any vertex $u_g$ in $A_2$ must be $d(u_g) > g$. Similarly, the Hamiltonian path that starts at $x_1$ shows that the degree of any vertex $v_h$ in $B_2$ is $d(v_h) > h$. Due to Theorem \ref{thm2}, $G$ is a yes instance of the Hamiltonian cycle problem. Therefore, $G$ is Hamiltonian.     
	\end{proof}
	\noindent
	{\bf 3. Exactly-2-Simple Path Cover}
	\\A simple path cover is a simple path that covers all the vertices of $G$, and the Hamiltonian path is one such example. A connected graph $G$ is said to have Exactly-2-Simple Path Cover \cite{nguyen2018various} if $V(G)$ can be covered by two simple paths but cannot be covered by one simple path. As the Hamiltonian path problem is a special case of this problem, the Exactly-2-Simple Path Cover problem is also NP-complete on graphs, whereas the Hamiltonian path problem is known to be NP-complete.   However, one can take up a complexity study of this problem on graphs where the Hamiltonian path problem is polynomial-time solvable. In this paper, we shall study in $P_5$-free chordal bipartite graphs. The notation $\exists!z R(z)$  refers to there exists unique $z$ that satisfies the predicate $R$.\\
	\begin{theorem}
		For a $P_5$-free chordal bipartite graph $G=(A=A_1\cup A_2, B=B_1\cup B_2)$ with $|A_2|\geq 1$ and $|B_2|\geq 2$, $G$ has Exactly-2-Simple Path Cover if and only if one of the following is true \\
		(i)  $|A|=|B|$ and $A_2$ has an ordering, $\exists!z_r, ~d(z_r) < r$ and $\forall u_g\ne z_r, d(u_g) \geq g$, $1 \leq g \leq p$ and $B_{2}$ has an ordering, $ \forall v_h, d(v_h) \geq h$, $1 \leq h \leq q$. \\
		(ii) $|A|=|B|$ and $A_2$ has an ordering, $\forall u_g\, d(u_g) \geq g$, $1 \leq g \leq p$ and $B_{2}$ has an ordering, $\exists! z_r, ~d(z_r) < r$ and $ \forall v_h\ne z_r, d(v_h) \geq h$, $1 \leq h \leq q$. \\
		(iii) $|A|=|B|+1$ and $A_2$ has an ordering, $\exists! z_r, ~d(z_r) < r$ and $\forall u_g\ne z_r, d(u_g) \geq g$, $1 \leq g \leq p$ and $B_{2}$ has an ordering, $ \forall v_h, d(v_h) >h$, $1 \leq h \leq q$. \\
		(iv) $|A|=|B|+1$ and $A_2$ has an ordering, $\forall u_g,, d(u_g) \geq g$, $1 \leq g \leq p$ and $B_{2}$ has an ordering, $\exists! z_r, ~d(z_r) \leq r$ and $ \forall v_h\ne z_r, d(v_h) > h$, $1 \leq h \leq q$. 	
	\end{theorem}
	\begin{proof}
		{\em Necessity:} (i) Without loss of generality, on the contrary, assume that there does not exist a vertex $z_r\in A_2$ ($z_h\in B_2$) such that $d(z_r)<r~ (d(z_h)<h)$ or there exist at least two vertices $z_r, z_s\in A_2$ such that $d(z_r)<r$ and $d(z_s)<s$. 
		
		\textbf{Case 1:} There does not exist $z_r$ ($z_h\in B_2$) such that $d(z_r)<r~ (d(z_h)<h)$. By Theorem \ref{thm3}, we observe that $G$ is a yes instance of the Hamiltonian path problem, which is a one simple path cover, a contradiction. 
		
		\textbf{Case 2:} There exist two vertices that violates \emph{NNO} property. Assume that the vertices $z_r$ and $z_s$  are the first two vertices such that $d(z_r)<r$ and $d(z_s)<s$ with respect to \emph{NNO}. Without loss of generality, assume that $d(z_r) \leq d(z_s)$. Let $\{u_1,\ldots, u_r=z_r, \ldots, u_s=z_s\,\ldots,u_p\}$ be the ordering of vertices in $A_2$.  By \emph{NNO} property, we know that $d(z_{r-1})=d(z_r)$ and $N_G(z_{r-1})=N_G(z_{r})$. Similarly, $d(z_{s-1})=d(z_s)$ and $N_G(z_{s-1})=N_G(z_{s})$. By \emph{NNO} property, we observe that $z_r$ ($z_s$) is the end vertex of any maximum path. Since $z_r$ and $z_s$ do not satisfy Theorem \ref{thm3}, these vertices cannot be in one simple path that covers $V(G)$. This shows that there exist at least two simple paths $P^1, ~P^2$ that cover $V(G)$. Let $P^1=(\ldots,u_{r}=z_r)$, $1<r<s$ and $P^2=(\ldots,u_{s}=z_s)$ $1<s\leq p$. Due to \emph{NNO}, $V(P^1)\cap A_2 \subseteq \{u_1,\ldots,u_r=z_r\}$ and $V(P^2)\cap A_2\subseteq  \{u_{r+1},\ldots,u_s=z_s\}$. Now, we analyze the following cases. 
		
		\textbf{Case 2.1:} $u_p\ne z_s$. Observe that the vertices $\{u_{s+1}=z{s+1},\ldots,u_p\}\cup \{y_s,\ldots,y_j\}\cup A_1\cup B_2$ cannot be a part of $P^1$ or $P^2$ and to cover these vertices we need another simple path $P^3= (u_{s+1}, y_{s}, \dots,  u_{p},y_{p},x_{q+1},y_{p+1},x_{q+2},y_{p+2},\ldots,x_{i}, y_{j},x_{q}, v_{q}, \ldots,x_{2},v_{2},x_{1}, v_{1})$. 
		
		\textbf{Case 2.2:} $u_p= z_s$. The following vertices $\{y_s,\ldots,y_j\}\cup A_1\cup B_2$ are not covered by $P^1$ and $P^2$ and these vertices are by another simple path $P^3= (y_{p},x_{q+1},y_{p+1},x_{q+2},y_{p+2},\ldots,x_{i}, y_{j},x_{q}, v_{q}, \ldots,x_{2},v_{2},x_{1}, v_{1})$. It is clear from both cases that $V(G)$ cannot be covered in two simple paths, a contradiction. A similar argument can be given for (ii), (iii), and (iv) of our claim.
		
		{\em Sufficiency:} There exists exactly one vertex $z_r$, such that $~d(z_r) < r$. By Theorem \ref{thm3}, $G$ is a no instance of the Hamiltonian path problem, and thus $G$ cannot be covered by one simple path. Since Theorem \ref{thm3} is constructive, we obtain the following two simple paths $P^1$ and $P^2$ that cover $V(G)$.  \\
		$P^1=(u_{1},y_{1},u_{2},y_{2},\ldots,u_{r-1},y_{r-1},u_{r}=z_r)$,  $1<r\leq p$\\
		$P^2=(y_{r}, u_{r+1}, y_{r+1},u_{r+2}\ldots,y_{p-1},u_{p},y_{p},x_{q+1},y_{p+1},x_{q+2},y_{p+2},\ldots,x_{i}, y_{j},x_{q}, v_{q}, \ldots,$ $x_{2},v_{2},x_{1},v_{1})$\\
		Similarly, (ii), (iii), and (iv) of our claim can be proved. 
	\end{proof}
	\noindent
	{\bf 4. Path-hypo-Hamiltonian} \\
	Recall that a hypo-Hamiltonian graph is a non-Hamiltonian graph $G$ such that $G - v$ is Hamiltonian for every vertex $v \in V (G)$. Clearly, bipartite graphs are non-hypo-Hamiltonian graphs. So, it is natural to ask for the variants of hypo-Hamiltonian. One such variant is path-hypo-Hamiltonian. A graph $G$ is called path-hypo-Hamiltonian if $G$ has no Hamiltonian path and $\forall v \in V(G)$, $G - v$ has a Hamiltonian path. In this section, we show that $P_5$-free chordal bipartite bisplit graphs are non-path-hypohamiltonian.\\
	\begin{theorem}
		Let $G$ be a $P_5$-free chordal bipartite bisplit graph without a Hamiltonian path. Then, $G$ is non-path-hypohamiltonian.
	\end{theorem}
	\begin{proof}
		We now exhibit a vertex $u$ such that $G-u$ has no Hamiltonian path. Since $G$ is a no instance of Hamiltonian path problem, by Theorem \ref{thm3}, either (i) or (ii) is true. (i) $|A|=|B|$ and there exists a vertex $u_g\in A_2$ such that $d(u_g)<g$ or $v_r\in B_2$ such that $d(v_r)<r$ (ii) $|A|=|B|+1$ and there exists a vertex $u_g\in A_2$ such that $d(u_g)<g$ or $v_r\in B_2$ such that $d(v_r)\leq r$.
		
		(i) Suppose that $G-u_g$ is a yes instance of the Hamiltonian path problem. Now we choose $v \in B_2$. In $G-v$, the degree $d_G(u_g)=d_{G-v}(u_g)$. By Theorem \ref{thm3},  $G-v$ is a no instance of the Hamiltonian path problem. Therefore, $G$ is non-path-hypo-Hamiltonian.
		
		(ii) We choose $v \in B_1 (B_2)$. In $G-v$, It is easy to see that $|A|>|B|+1$. By Theorem \ref{thm3}, $G-v$ is a no instance of Hamiltonian path problem. Therefore, $G$ is non-path-hypo-Hamiltonian.
	\end{proof}
	\noindent{\bf Remark:} The problems discussed in this section and their proofs are constructive in nature. Therefore, we obtain polynomial-time algorithms for Hamiltonicity variants.
	\subsection{Generalizations of Hamiltonian path(cycle) in $P_5$-free Chordal Bipartite Graphs}
	\label{general}
	It is known that the longest path (cycle) problem is NP-hard on every class of graphs on which HPATH (HCYCLE) is NP-complete. Thus, if someone is interested in investigating the tractability of the longest path (cycle) problem, a generalization of HPATH (HCYCLE), it makes sense to focus on the classes of graphs for which HPATH (HCYCLE) is polynomial-time solvable. It is important to highlight that there are very few known polynomial time algorithms for the longest path (cycle) problem. Interestingly, polynomial algorithms have been proposed that solve the longest path problem on bipartite permutation graphs \cite{uehara2007computing},  interval graphs \cite{ioannidou2011longest}, and cocomparability graphs \cite{mertzios2012simple}. Having obtained a polynomial-time algorithm for HPATH (HYCLE) in $P_5$-free chordal bipartite graphs, it is natural to explore the complexity status of the longest path (cycle) problem in $P_5$-free chordal bipartite graphs. In this section, we use HPATH (HCYCLE) algorithm as a framework to solve other combinatorial problems. For all other problems, we modify the input graph to obtain  $G^*$. The challenge lies in identifying $G^*$ for each combinatorial problem, such as longest cycle and longest path. By calling the appropriate algorithm (Hamiltonian cycle or Hamiltonian path Algorithm) on $G^*$, we obtain a partial result. A suitable modification to our partial result will give us the longest path (cycle), for $G$. We shall see some of the generalizations of Hamiltonicity.
	
	\subsubsection{Longest paths in $P_5$-free Chordal Bipartite Graphs}
	For a connected graph $G$, the longest path is a path of maximum length in $G$.  In this section, we shall investigate the complexity of the longest path problem in $P_{5}$-free chordal bipartite graphs.  \\\\
	{\bf Pruning:} We shall now prune $G$ by removing vertices that will not be part of any longest path in $G$.  Without loss of generality, we assume that $G$ has no Hamiltonian path, and hence $|A|=|B|+1+f$, $f \geq 1$ or there must exist a vertex in $A_2$ ($B_2$) that does not satisfy the conditions mentioned in Theorem \ref{thm3}.  As part of pruning, we prune such vertices from $G$.   Recall that $A_2=(u_{1},u_{2},\ldots,u_{p})$.  Let  $u_r$ be the first vertex in $A_2$ with $d(u_{r})<r$.  Remove $u_r$ and relabel the vertices of $A_2$ so that the sequence is reduced to $(u_{1},u_{2},\ldots,u_{p-1})$.  With respect to the modified sequence, if we find $u_i$ such that $d(u_{i}) < i$, then prune $u_i$ and update the sequence.  If there are no such $u_r$, then $c=0$.  After, say $c$ iterations, $A_2$ becomes $(u_{1},u_{2},\ldots,u_{p-c})$ such that for $\forall u_g, 1 \leq g \leq (p-c), d(u_g) \geq g$.  Similarly, after $d$ iterations, $B_2$ becomes $(v_{1},v_{2},\ldots,v_{q-d})$ such that for $\forall v_h, 1 \leq h \leq (q-d), d(v_h) \geq h$.   After pruning the vertices in $A$ is reduced to the set $A'$, $|A'|= |A| - c$. Similarly, $B$ is reduced to the set $B'$, $|B'|=|B|-d$.  From now on, when we refer to $A_2$ ($B_2$), it refers to the modified $A_2$ ($B_2$).  Let $G^*$ be the modified graph of $G$. $V(G^*) = V(A') \cup V(B')$. \\
	
	\textbf{Case 1: $|A'| = |B'|$}.	We observe that $G^*$ satisfies {\em NNO} and as per Theorem \ref{thm3}, $G^*$ has a Hamiltonian path, which is $ P= (u_1, y_1, u_2, y_2, \ldots, u_{p-c}, y_{p-c}, x_{(q-d)+1}, y_{(p-c)+1}, $ $x_{(q-d)+2}, y_{(p-c)+2}, \ldots, x_i, y_j, x_{q-d}, v_{q-d},$ $x_{(q-d)-1},$ $v_{(q-d)-1},$  $ \ldots, x_1, v_1)$  
	
	\textbf{Case 2: $|A'| = |B'|+1+f$, $f \geq 0 $}. If $f>0$, by Theorem \ref{thm3}, $G^*$ is not a yes instance of the Hamiltonian path problem.  We remove $f$ vertices from $A'_2$.  Let $G^*_1(A'\setminus\{u_{(p-c)},u_{(p-c)-1},u_{(p-c)-2},\ldots, u_{(p-c)-(f-1)}\},B')$ be the modified graph.  If $f=0$, then $G^*_1(A',B')$ be same as $G^*$. Further, if $f>|A'_2|$, then we remove $\{x_i,\ldots,x_{i+(p-c)-f+1}\}$ from $A'_1$. 
	
	\textbf{Case 2.1:} $\exists  v_r$ $\in B'_2 $ in $G^*_1$ such that $d_{G^*_1}(v_{r}) = r$. Since $d_{G^*_1}(v_{r}) = r$, $G^*_1$ is not a yes instance of the Hamiltonian path problem as per Condition (ii) of Theorem \ref{thm3}. We remove the vertex $u_{(p-c)-(f)}$ from $G^*_1$ to obtain $|A'|=|B'|$. Clearly, $G^*_1(A' \setminus \{u_{(p-c)},u_{(p-c)-1},u_{(p-c)-2},\ldots, u_{(p-c)-(f-1)},u_{(p-c)-(f)}\},B')$ has a Hamiltonian path as per Condition (i) of Theorem \ref{thm3}.\\
	$ P= (u_1, y_1, u_2, y_2, \ldots, u_{(p-c)-(f+1)}, y_{(p-c)-(f+1)}, x_{(q-d)+1},y_{((p-c)-(f+1))+1},x_{(q-d)+2},\\ y_{((p-c)-(f+1))+2}, \ldots, x_i,$ $ y_j, x_{q-d}, v_{q-d},$ $ x_{(q-d)-1},$ $ v_{(q-d)-1}, \ldots,$ $ x_1, v_1)$  \
	
	\textbf{Case 2.2:} $\nexists  v_r$ $\in B'_2 $ in $G^*_1$ such that $d_{G^*_1}(v_{r}) = r$. By Theorem \ref{thm3}, $G^*_1$ is a yes instance of the Hamiltonian path problem.\\
	$ P= (u_1, y_1, u_2, y_2, \ldots, u_{(p-c)-f}, y_{(p-c)-f}, x_{(q-d)+2},y_{((p-c)-f)+1},x_{(q-d)+3}, y_{((p-c)-f)+2}, \\\ldots, x_i, y_j, x_{(q-d)+1}, v_{q-d}, x_{(q-d)},$ $ v_{(q-d)-1}, \ldots,$ $ x_2, v_1, x_1)$  \\ 
	
	\noindent\textbf{Claim 1:} $P$ is a longest path in $G$. 
	\begin{proof}
		Let $G^*$ be the graph obtained by pruning the vertices from $A_2$ and $B_2$ in $G$, and thus $G^*$ becomes a yes instance of the Hamiltonian path problem.  Since $G^*$ satisfies {\em (NNO)}, for all the pruned vertices of $A_2$ and $B_2$, their neighborhood is a subset of $\{y_1,\ldots,y_{p-c}\}$ and $\{x_1,\ldots,x_{q-d}\}$ respectively. This shows that the pruned vertices of $A_2$ and $B_2$ cannot be augmented to $P$ to get a longer path in $G$.  As $|i-j|$ is minimum the choice of $|i-j|$ ensures the number of vertices pruned is minimum.  This proves that $P$ is a maximum path in $G$.  We now argue that $P$ is also the longest path.  Let $S$ denote the set of all pruned vertices.  Suppose there exists a path $P'$ which is longer than $P$.  It must be the case that $P'$ is completely different from $P$.  We know that the vertices of $P'$ violate degree constraints mentioned in Theorem \ref{thm3}.  It is important to highlight that the violated vertices cannot be used to obtain a path as it follows \emph{NNO}; otherwise, these vertices can be augmented to $P$ to get a longer path, which contradicts the maximality of $P$.  Hence, the path obtained $P$ is the longest path.
	\end{proof}
	\noindent\textbf{Trace of the longest path algorithm:}
	\begin{figure}
		\begin{center}
			\includegraphics[width=80mm,scale=0.5]{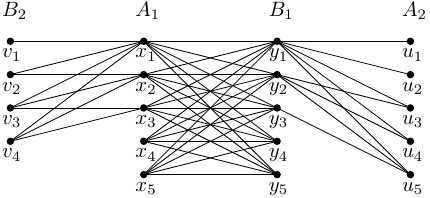}
			\caption{\small \sl   An illustration for the proof of Longest path (Case 1).\label{fig:Lp1}}
		\end{center}
	\end{figure}
	For Figure \ref{fig:Lp1}, we shall trace the longest path algorithm.  Consider the vertices of $A_2$,  $d(u_2)=1$, $u_2$ violates the degree constraint, so we prune $u_2$ and relabel the vertices $u_3$ as $u_2$, $u_4$ as $u_3$, and $u_5$ as $u_4$.  The updated sequence of $A_2$ is $(u_1,u_2,u_3,u_4)$. 
	With respect to the modified sequence, $d_{G^*}(u_3)=2$, prune $u_3$ and relabel the vertex $u_4$ as $u_3$. Now all the vertices of $A'_2$ satisfy the degree constraint, and the sequence is $(u_1,u_2,u_3)$.  
	By applying the procedure to the vertices of $B_2$, we get $B'_2$ as $(v_1,v_2,v_3)$.  
	The resultant graph $G^*$ falls under Case 1.  We obtain the longest path $ P= (u_1, y_1, u_2, y_2, u_3, y_3, x_4, y_4,x_5, y_5,x_3,v_3,x_2,v_2,x_1,v_1)$  \\
	\begin{figure}
		\begin{center}
			\includegraphics[width=80mm,scale=0.5]{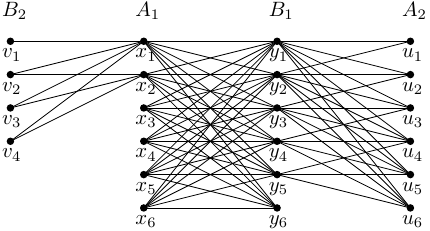}
			\caption{\small \sl  An illustration for the proof of Longest path (Case 2).\label{fig:Lp2}}
		\end{center}
	\end{figure}
	
	 Now consider Figure \ref{fig:Lp2}.  The vertex $u_6$ violates the degree constrains, prune $u_6$ and we obtain $A'_2$, whose sequence reduced to $(u_1,u_2,u_3,u_4,u_5)$.  Similarly, $v_3$ and $v_4$ violate the constraint and thus the sequence of $B'_2$ is reduced to  $(v_1,v_2)$.  $G^*$ follows Case 2 of the procedure. Let $G^*_1$ be the graph obtained by removing $f$,$f=\{u_4,u_5\}$ vertices from $G^*$.  We observe that $d_{G^*}(v_1)=1$, by Case 2.1, remove $u_3$ and the longest path is  $ P= (u_1, y_1, u_2, y_2, x_3, y_3, x_4, y_4,x_5, y_5,x_6, y_6,x_2,v_2,x_1,v_1)$  \\ 
	\begin{theorem}
		\label{thmlong}
		Let $G$ be a $P_5$-free chordal bipartite graph.  Finding the longest path in $G$ is linear-time solvable.
	\end{theorem}
	\begin{proof}
		Follows from the above discussion on pruning and Claim 1. 
	\end{proof}
\noindent	{\bf Remarks:} As an extension of the longest path problem, we naturally obtain a minimum leaf spanning tree of $G$, which is a spanning tree of $G$ with the minimum number of leaves in linear time.  Since $G$ satisfies {\em {\em {\em NNO}}}, the vertices pruned while constructing $G^*$ cannot be included as internal vertices of $P$.  We shall now construct a minimum leaf spanning tree $T$ with $P$ as a subtree.  The pruned vertices are augmented to $P$ as leaves to obtain $T$.  This shows that $T$ has $|V(G) \setminus V(P)| + 2$ leaves.  Maximum leaf spanning tree of $G$ can be constructed by choosing $x_1y_1$ edge and augment all other vertices of $A_1, A_2$ to the vertex $y_1$, similarly $B_1, B_2$ to $x_1$. \\
\begin{corollary}
	The minimum connected dominating set in $P_5$-free chordal bipartite graph $G$ is linear-time solvable.
\end{corollary}
\begin{proof}
	Since $G$ has {\em {\em {\em NNO}}} property, $\{x_1,y_1\}$ is a trivial minimum connected dominating set in $G$. 
\end{proof}
	\subsubsection{Longest cycles in $P_5$-free Chordal Bipartite Graphs}
	Similar to the longest path, the longest cycle is an induced cycle of maximum length in $G$. We work with no instances of HCYCLE. It is clear that $A\ne B$ or there must exist a vertex in $A_2 ~(B_2)$ that does not satisfy the condition mentioned in Theorem 2.  We prune the violated vertices from $G$ that are not part of the longest cycle.  Let $u_r$ be the first vertex in $A_2$ with $d(u_r) \leq r$.  Remove $u_r$ and relabel the vertices of $A_2$ until there is no such  $u_r$.  After say $c$ iterations, $A_2$ becomes $(u_1, u_2,\ldots, u_{p-c})$ such that for $\forall$$u_g$, 1$\leq g \leq (p-c)$, $d(u_g) > g$. 
	Similarly, after say $d$ iterations, $B_2$ becomes $(v_1, v_2,\ldots, v_{q-d})$ such that for $\forall$$v_h$, 1$\leq h \leq (q-d)$, $d(u_h) > h$. After pruning, $A$ reduced to the set $A'$, $|A'|= |A|- c$ and $B$ reduced to the set $B'$, $|B'|=|B|-d$.  Let the modified graph be $G^*$.\\ 
	
	\noindent \textbf{Case 1: $|A'| = |B'|$}\\
	$G^*$ satisfies {\em NNO} and by Theorem \ref{thm2}, $G^*$ is a yes instance of the Hamiltonian cycle problem. \\
	$ C= (y_1,u_1, y_2,u_2,  \ldots, y_{p-c},u_{p-c}, y_{(p-c)+1}, x_1, v_1,x_2, v_2, \ldots,$ $  x_{q-d},$ $ v_{q-d}, x_{(q-d)+1},\\y_{(p-c)+2}, x_{(q-d)+2},  \ldots,$ $  y_j,$ $ x_i, y_1 )$ \\ 
	\noindent \textbf{Case 2: $|A'| = |B'|+f$, $f\geq 1$}\\
	By Theorem \ref{thm2}, $|A'| = |B'|$ and hence $G^*$ is a yes instance of the Hamiltonian cycle problem.  We remove $f$, $\{u_{(p-c)},u_{(p-c)-1},u_{(p-c)-2},\ldots, u_{(p-c)-(f-1)}\}$ vertices from $A'_2$ in $G^*$ results  $|A'| = |B'|$ and by Theorem \ref{thm2}, $G^*$ has a Hamiltonian cycle.
	\\ $ C= (y_1,u_1, y_2,u_2,  \ldots, y_{(p-c)-f},u_{(p-c)-f}, y_{((p-c)-f)+1}, x_1, v_1,x_2, v_2,$ $\ldots,x_{q-d}, v_{q-d},\\ x_{(q-d)+1},y_{((p-c)-f)+2},$ $ x_{(q-d)+2},  \ldots,$ $ y_j, x_i, y_1 ) $ \\
	
	\noindent \textbf{Claim 2:} $C$ is a longest cycle in $G$. 
	\begin{proof}
		We prune the vertices from $A_2$ and $B_2$ in $G$ and let the modified graph be $G^*$.  We observe that $G^*$ satisfies the condition mentioned in Theorem 2.  Since $G^*$ satisfies {\em (NNO)}, the neighborhood of pruned vertices of $A_2$ and $B_2$ is a subset of $\{y_1,\ldots,y_{p-c}\}$ and $\{x_1,\ldots,x_{q-d}\}$, respectively.  Therefore, the pruned vertices cannot be augmented to $C$ to get a longer cycle in $G^*$.  Hence $C$ is the longest cycle in $G^*$. 
	\end{proof}
	
	\noindent\textbf{Trace of the Algorithm:}
	Consider the graph $G$ given in Figure \ref{fig:Lp1}, $u_1$ and $u_2$ violates the degree constraint, we prune the vertices, and the sequence becomes $(u_1,u_2,u_3)$.  With respect to the modified sequence $d_{G^*}(u_2)=2$, prune $u_2$ and the sequence becomes $(u_1,u_2)$.  Similarly, in $B_2$, $v_1$ is pruned, $(v_1,v_2,v_3)$.  We find $v_3$ such that $d_{G^*}(v_2)=2$, prune $v_3$, $(v_1,v_2)$.  Graph $G^*$ satisfies case 1, $|A'| = |B'|$ and the longest cycle $C=(y_1,u_1, y_2,u_2,y_3, x_1, v_1,x_2, v_2,x_3,y_4,x_4,y_5,x_5,y_1)$. \\
	Case 2: Consider the graph $G$ given in Figure \ref{fig:Lp2}, vertices $u_5$ and $u_6$ are pruned and $A'$ becomes  $(u_1,u_2,u_3,u_4)$.  Similarly, $v_1$ is pruned and the sequence becomes $(v_1,v_2,v_3)$.  In the modified graph, $v_2$ and $v_3$ violate degree constraint, and thus $B'_2$ has a vertex $v_1$.  $G^*$ satisfies case 2 of Theorem \ref{thm3}, $|A'| = |B'|+f$, $f=3$.  Let $G^*_1$ be the graph obtained by removing $\{u_2,u_3,u_4\}$.  The longest cycle $C=(y_1,u_1, y_2, x_1, v_1,x_2,y_3,x_3,y_4,x_4,y_5,x_5,y_6,x_6,y_1)$. \\
	
	{\bf Remark:} Proofs of the problems discussed in this section are constructive in nature, and thus, we obtain a polynomial-time algorithm for a generalization of the  Hamiltonian cycle (path).\\\\
	\noindent
	\textbf{Time Complexity Analysis:}
	The {\em  Nested Neighbourhood Ordering (NNO)} of $P_5$-free chordal bipartite graphs can be obtained in linear time by using the perfect edge elimination ordering proposed by Fulkerson and Gross. Further, the Hamiltonian cycle (path) can be obtained in linear time. Since bipancyclic, and homogeneously traceable involves order $n$ call to Hamiltonian cycle (path) algorithm, these two incur $\mathit{O(n^2)}$ time. All other algorithmic results run in linear time.
	\section{Conclusions and Further Research}
	In this paper, we performed fine-grained analysis of HCYCLE (HPATH) for bisplit graphs with chordality as the parameter and proved that HCYCLE (HPATH) is polynomial-time solvable on chordal bisplit graphs and NP-complete on chordal bipartite bisplit graphs. Further, we proved HCYCLE is NP-complete on $P_{10}$-free chordal bipartite graphs. On the algorithmic front, we presented a polynomial-time algorithm for HCYCLE (HPATH) on $P_5$-free chordal bipartite graphs, using which we solved its variants and generalizations.   A natural direction for further research is to reduce the complexity gap between $P_5$-free chordal bipartite graphs and $P_{10}$-free chordal bipartite graphs for HCYCLE (HPATH).  \\\\
		\textbf{Declaration} - The authors declare that they have no conflict of interest.
	
	\bibliography{references}
\end{document}